\documentclass[11pt, a4paper]{article}
\usepackage{amssymb, amsmath, amsthm}
\usepackage{graphicx}
\newtheorem{theorem}{Theorem}[section]
\newtheorem{corollary}[theorem]{Corollary}
\newtheorem{lemma}[theorem]{Lemma}

\theoremstyle{definition}
\newtheorem{definition}[theorem]{Definition}
\theoremstyle{remark} \theoremstyle{remark}

\numberwithin{equation}{section}

\makeatletter 

\newcounter{robA}
\newcounter{robB}
\newcounter{robC}
\newcounter{robD}
\newcounter{robE}
\newcounter{robF}
\newcounter{robG}
\newcounter{robTrzyA}
\newcounter{robTrzyB}
\newcounter{robCztA}
\newcounter{robCztB}
\newcounter{robPiecA}
\newcounter{robPiecB}

\def\trojkat(#1,#2,#3){
    \setcounter{robA}{#3}
    \divide\c@robA2
    \put(#1,#2){\line(1,0){#3}}
    \put(#1,#2){\line(2,3){\value{robA}}}
    \setcounter{robB}{#1}
    \addtocounter{robB}{#3}
    \put(\value{robB},#2){\line(-2,3){\value{robA}}}
}

\def\trojkatABC(#1,#2,#3){
    \trojkat(#1,#2,#3)
    \setcounter{robA}{#1}
    \addtocounter{robA}{-8}
    \setcounter{robB}{#2}
    \addtocounter{robB}{-5}
    \put(\value{robA},\value{robB}){$a$}
    \setcounter{robC}{#3}
    \addtocounter{robC}{3}
    \addtocounter{robC}{#1}
    \put(\value{robC},\value{robB}){$b$}
    \addtocounter{robA}{10}
    \setcounter{robD}{#3}
    \divide\c@robD2
    \addtocounter{robA}{\value{robD}}
    \divide\c@robD2
    \multiply\c@robD3
    \addtocounter{robD}{2}
    \addtocounter{robD}{#2}
    \put(\value{robA},\value{robD}){$c$}
}

\def\trojkatLevelDwa(#1,#2,#3){
\trojkat(#1,#2,#3)
\setcounter{robF}{#1}%
\addtocounter{robF}{#3}%
\trojkat(\value{robF},#2,#3)
\setcounter{robE}{#3}
\divide\c@robE2
\addtocounter{robE}{#1}
\setcounter{robG}{#3}
\multiply\c@robG3
\divide\c@robG4
\addtocounter{robG}{#2}
\trojkat(\value{robE},\value{robG},#3)
}

\def\trojkatLevelTrzy(#1,#2,#3){
\trojkatLevelDwa(#1,#2,#3)
\setcounter{robTrzyA}{#1}
\addtocounter{robTrzyA}{#3}
\addtocounter{robTrzyA}{#3}
\trojkatLevelDwa(\value{robTrzyA},#2,#3)
\setcounter{robTrzyB}{#3}
\multiply\c@robTrzyB3
\divide\c@robTrzyB2
\addtocounter{robTrzyB}{#2}
\addtocounter{robTrzyA}{-#3}
\trojkatLevelDwa(\value{robTrzyA},\value{robTrzyB},#3)
}

\def\trojkatLevelCztery(#1,#2,#3){
\trojkatLevelTrzy(#1,#2,#3)
\setcounter{robCztA}{#3}
\multiply\c@robCztA4
\addtocounter{robCztA}{#1}
\trojkatLevelTrzy(\value{robCztA},#2,#3)
\setcounter{robCztB}{#3}
\multiply\c@robCztB3
\addtocounter{robCztB}{#2}
\addtocounter{robCztA}{-#3}
\addtocounter{robCztA}{-#3}
\trojkatLevelTrzy(\value{robCztA},\value{robCztB},#3)
}

\def\trojkatLevelPiec(#1,#2,#3){
\trojkatLevelCztery(#1,#2,#3)
\setcounter{robPiecA}{#3}
\multiply\c@robPiecA8
\addtocounter{robPiecA}{#1}
\trojkatLevelCztery(\value{robPiecA},#2,#3)
\setcounter{robPiecB}{#3}
\multiply\c@robPiecB6
\addtocounter{robPiecB}{#2}
\addtocounter{robPiecA}{-#3}
\addtocounter{robPiecA}{-#3}
\addtocounter{robPiecA}{-#3}
\addtocounter{robPiecA}{-#3}
\trojkatLevelCztery(\value{robPiecA},\value{robPiecB},#3)
}

\makeatother 

\begin{document}

\title{Hierarchical Random Graphs Based  on Motifs}
\author{Monika Kotorowicz\thanks{Instytut Matematyki, Uniwersytet Marii Curie-Sk{\l}odwskiej,
     20-031 Lublin, Poland, e-mail:
        monika@hektor.umcs.lublin.pl} \and Yuri Kozitsky\thanks{Instytut Matematyki, Uniwersytet Marii Curie-Sk{\l}odwskiej,
     20-031 Lublin, Poland, e-mail:
        jkozi@hektor.umcs.lublin.pl}}
\maketitle

\begin{abstract}
Network motifs are characteristic patterns which occur in the
networks essentially more frequently than the other patterns. For
five motifs found in S. Itzkovitz, U. Alon, Phys. Rev.~E, 2005, 71,
026117-1, hierarchical random graph models are proposed in which the
motifs appear at each hierarchical level. A rigorous construction of
such graphs is performed and a number of their structural properties
are analyzed. This includes degree distribution, amenability,
clustering, and the small world property. For one of the motifs,
annealed phase transitions in the Ising model based on the
corresponding graph  are also studied.

\end{abstract}

\section{Introduction}
 In view of the complexity and unknown organizing
principles of large real-world networks, they usually are modeled by
means of random graphs,  the study of which traces back to P.
Erd\H{o}s and A. R\'enyi \cite{ErdosRenyi1960}. Many of such
networks contain characteristic patterns recurring essentially more
frequently than the other ones. These are \textit{network motifs}
\cite{Itzkovitz2003,ItzkovitzAlon2005,Matias,Milo2002}. Quite often
real networks are build up mostly of motifs, which thus can be
treated as constructing units for their modeling, cf.
\cite{Milo2002}. In \cite{ItzkovitzAlon2005}, the authors introduced
a random graph model based on some geometric principles
(constraints). Then they compared the appearance of eight elementary
three- and four-node patterns in their model with the same
characteristics of the Erd\H{o}s-R\'enyi random graph. It turned out
that five of these patterns are motifs for their model, but not for
the Erd\H{o}s-R\'enyi random graph, see Fig.~\ref{motifs}.

\begin{figure}[htbp]
\centering
\begin{picture}(300, 45)
\trojkat(0,0,50)
\trojkat(70,0,50)
\put(95,38){\line(1,0){25}}
\put(140,0){\framebox(37,37)}
\put(200,0){\framebox(37,37)}
\put(260,0){\framebox(37,37)}
\put(200,0){\line(1,1){37}}
\put(260,0){\line(1,1){37}}
\put(297,0){\line(-1,1){37}}
\end{picture}
\caption{Three and four node motifs $M_1, M_2, M_3, M_4, M_5$ found in \cite{ItzkovitzAlon2005}.} \label{motifs}
\end{figure}
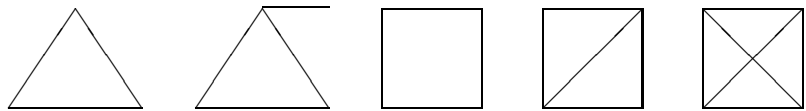

One of the ways to get information about infinite graphs is to study
the properties of certain models of statistical physics defined
thereon. The most popular ones are the Ising and Potts models, see
\cite{H}. On the other hand, graphs are employed to mimic crystal
lattices. For graphs with special structure, the critical behavior
of the Ising model can be described in an explicit and rigorous way.
This, in particular, holds for  the so-called hierarchical lattices
introduced in \cite{BO,GriffithsKaufman1982}. Such lattices are
constructed in an algorithmic way by means of basic patterns, e.g.,
by a `diamond', see $M_3$ in Fig.~\ref{motifs}. The relative
simplicity of the theory makes hierarchical lattices attractive in
studying critical point behavior of various types, see quite recent
works \cite{Ant,Aral} and the references therein. A mathematical
description of the Gibbs states of the Ising model on such graphs
was done by P.M. Bleher and E. \v{Z}alys in \cite{BleherZalys1988,
BleherZalys1989}. M. Hinczewski and A. Nihat Berker
\cite{HinczBerker2006} studied the critical point properties of the
Ising model on the diamond hierarchical lattice `decorated' by
adding random bonds. In the present paper, we follow the way
suggested in \cite{HinczBerker2006} and introduce hierarchical
graphs constructed by means of the motifs shown in
Fig.~\ref{motifs}, decorated by random bonds which somehow repeat
the corresponding motif. We analyze a number of their
characteristics, such as the average degree, the node degree
distribution, amenability, the small-world property. We also study a
ferromagnetic phase transition in the Ising model defined on the
graph based on $M_1$. A preliminary study of the models introduced
here was performed in \cite{Monika,Wr}.

\section{The Graphs: Construction and Structural Properties}
\subsection{The construction: informal description}

As is typical for hierarchical graphs, e.g., for hierarchical
lattices in \cite{GriffithsKaufman1982,HinczBerker2006}, the
construction is carried out in an algorithmic way: at $k$-th level,
$k\in \mathbb{N}$, one produces a subgraph, say $\Lambda_k$, which
is  then used as a construction element for producing
$\Lambda_{k+1}$. The procedure is the same at each level. The
starting element at level $1$ is obtained from the corresponding
motif. Let us illustrate this in the simplest case based on $M_1$ --
the triangle. To obtain $\Lambda_1$, we label the nodes of $M_1$ by
$a$, $b$, and $c$, as shown in Fig. \ref{unformalConstr}. The graph
$\Lambda_2$ is created in two step. First we take three graphs of
level $1$ and label them by $\Lambda_1^a$, $\Lambda_1^b$ and
$\Lambda_1^c$. Thereafter, the triangles are being glued up
according to the following rule: for $i,j \in \{a,b,c\}$, $i \neq
j$, node $i$ of triangle $\Lambda_1^j$ is glued up with node $j$ of
triangle $\Lambda_1^i$. The nodes $i$ od triangle $\Lambda_1^i$
remain untouched. These are the {\it external} nodes of $\Lambda_2$.
The remaining nodes are called {\it internal}. The bonds of the
initial triangles $\Lambda_1^i$, $i \in \{a,b,c\}$ turn into the
bonds of $\Lambda_2$. We call them {\it basic} bonds; they are
depicted as solid lines. At the second step, we add bonds connecting
the external nodes in the same way as it is in the motif $M_1$. Such
bonds are depicted as dotted lines and called {\it decorations}. As
a result, we obtain the graph $\Lambda_2$, which has nine basic bonds
and three decorations, three external and three internal nodes.

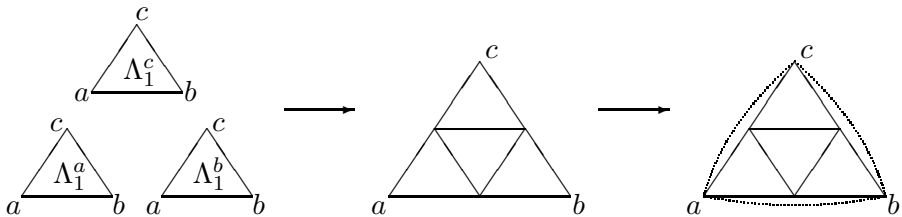
\begin{figure}[htbp]
\centering
\unitlength 0,23mm
\begin{picture}(500, 120)
    \put(10,0){\line(1,0){52}}
    \put(10,0){\line(2,3){26}}
    \put(62,0){\line(-2,3){26}}
    \put(1,-10){$a$}
    \put(62,-10){$b$}
    \put(26,38){$c$}
    \put(28,8){$\Lambda_1^a$}
    \put(90,0){\line(1,0){52}}
    \put(90,0){\line(2,3){26}}
    \put(142,0){\line(-2,3){26}}
    \put(81,-10){$a$}
    \put(142,-10){$b$}
    \put(119,38){$c$}
    \put(108,8){$\Lambda_1^b$}
    \put(50,60){\line(1,0){52}}
    \put(50,60){\line(2,3){26}}
    \put(102,60){\line(-2,3){26}}
    \put(40,55){$a$}
    \put(103,55){$b$}
    \put(75,102){$c$}
    \put(68,68){$\Lambda_1^c$}
\put(160,50){\vector(1,0){40}}
    \put(220,0){\line(1,0){52}}
    \put(220,0){\line(2,3){26}}
    \put(272,0){\line(-2,3){26}}
    \put(210,-10){$a$}
    \put(272,0){\line(1,0){52}}
    \put(272,0){\line(2,3){26}}
    \put(324,0){\line(-2,3){26}}
    \put(325,-10){$b$}
    \put(246,39){\line(1,0){52}}
    \put(246,39){\line(2,3){26}}
    \put(298,39){\line(-2,3){26}}
    \put(275,80){$c$}
\put(340,50){\vector(1,0){40}}
    \put(400,0){\line(1,0){52}}
    \put(400,0){\line(2,3){26}}
    \put(452,0){\line(-2,3){26}}
    \put(390,-10){$a$}
    \put(452,0){\line(1,0){52}}
    \put(452,0){\line(2,3){26}}
    \put(504,0){\line(-2,3){26}}
    \put(505,-10){$b$}
    \put(426,39){\line(1,0){52}}
    \put(426,39){\line(2,3){26}}
    \put(478,39){\line(-2,3){26}}
    \put(455,80){$c$}
   \qbezier[50](400,0)(410,40)(452,78)
   \qbezier[50](400,0)(453,-10)(504,0)
    \qbezier[50](452,78)(494,40)(504,0)
\end{picture}
\caption{Construction of the graph $\Lambda_2$ based on $M_1$ } \label{unformalConstr}
\end{figure}

To obtain $\Lambda_k$, $k=3,4,\ldots$, we repeat the same procedure
- take three copies of $\Lambda_{k-1}$ and label them by
$\Lambda_{k-1}^a$, $\Lambda_{k-1}^b$, and $\Lambda_{k-1}^c$. Then
the graphs $\Lambda_{k-1}^i$, $i \in \{a,b,c\}$ are glued up as
described above. Thereafter, three decorating bonds are drawn to
connect the external nodes. This procedure is repeated ad infinitum.

\subsection{Definitions}

In this subsection we begin performing the mathematical construction
of the model outlined above. In order to fix the terminology, we
recall relevant mathematical notions. A simple graph ${\sf G}$ is a
pair of sets $({\sf V}, {\sf E})$, where ${\sf V}$ is the set of
nodes, whereas ${\sf E}$ is a subset of the Cartesian product ${\sf
V} \times {\sf V}$. It is symmetric and irreflexive, i.e., $\langle
j, i \rangle \in {\sf E}$ whenever $\langle  i ,j\rangle \in {\sf
E}$, and  $\langle i, i \rangle \notin {\sf E}$ for every $i, j\in
{\sf V}$. We say that $i$ and $j$ are connected by a {\it bond} if
$\langle i, j \rangle \in {\sf E}$. In this case, we write $i \sim
j$ and say that $i$ and $j$ are {\it adjacent} or that they are {\it
neighbors}. Hence, the elements of ${\sf E}$ themselves can be
called bonds. The graph is said to be {\it complete}, if each two
nodes are adjacent. For a given $i$, by $n(i)$ we denote the {\it
degree} of $i$ -- the number of its neighbors. If ${\sf V}$, and
hence  ${\sf E}$, are finite, the graph is said to be finite.
Otherwise, the graph is infinite. An infinite graph is called {\it
locally finite}, if $n(i)$ is finite for every node.

Given ${\sf G}=({\sf V}, {\sf E})$ and ${\sf G}'=({\sf V}', {\sf
E}')$, let $\phi: {\sf V} \rightarrow {\sf V}'$ be such that $\phi
(i) \sim \phi(j)$ whenever $i\sim j$. Such a map $\phi$ is called a
{\it morphism}. A bijective morphism is called an {\it isomorphism}.
If $\phi$ is an isomorphism, then its inverse $\phi^{-1}$ is also an
isomorphism, and then the graphs ${\sf G}$ and ${\sf G}'$ are said
to be isomorphic. Such graphs have identical structures and thus can
be identified. In this case, we also say that ${\sf G}'$ is a {\it
copy} of ${\sf G}$. One observes that this refers to both finite and
infinite graphs. An isomorphism $\phi : {\sf V} \rightarrow {\sf
V}$, i.e. which maps the graph onto itself, is called an {\it
automorphism}. The graph ${\sf G}'=({\sf V}', {\sf E}')$ such that
${\sf V}' \subset {\sf V}$ and ${\sf E}' \subset {\sf E}$ is said to
be a {\it subgraph} of ${\sf G}=({\sf V}, {\sf E})$. In this case,
we write ${\sf G}' \subset {\sf G}$. Suppose that a subgraph ${\sf
G}' \subset {\sf G}$ has a copy, say ${\sf G}''$, that is, there
exists an isomorphism $\phi : {\sf G}'' \rightarrow {\sf G}'$. Then
$\phi$, considered as a map $\phi : {\sf G}'' \rightarrow {\sf G}$,
is called an {\it embedding} of ${\sf G}''$ into ${\sf G}$, whereas
${\sf G}'$ is called the {\it image} of ${\sf G}''$ under this
embedding. Fig.~\ref{motifs} presents the so called {\it unlabeled}
graphs, which are studied in this work. After labeling, i.e.,
attaching a label to each of the nodes, such a pattern turns into a
graph. Another labeling may or may not give the same graph up to an
automorphism. This depends on whether or not there exists the
corresponding automorphism.

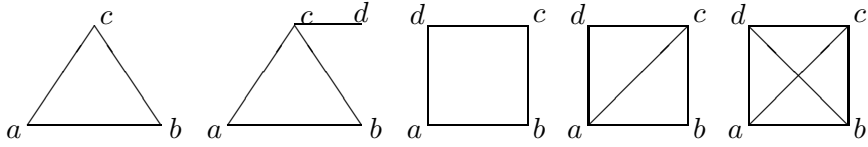
\begin{figure}[htbp]
\centering
\begin{picture}(300, 45)
\trojkatABC(0,0,50)
\trojkatABC(75,0,50)
\put(100,38){\line(1,0){25}}
\put(122,39){$d$}
\put(150,0){\framebox(37,37)}
\put(142,-5){$a$}
\put(189,-5){$b$}
\put(143,38){$d$}
\put(189,39){$c$}
\put(210,0){\framebox(37,37)}
\put(202,-5){$a$}
\put(249,-5){$b$}
\put(203,38){$d$}
\put(249,39){$c$}
\put(270,0){\framebox(37,37)}
\put(262,-5){$a$}
\put(309,-5){$b$}
\put(263,38){$d$}
\put(309,39){$c$}
\put(210,0){\line(1,1){37}}
\put(270,0){\line(1,1){37}}
\put(307,0){\line(-1,1){37}}
\end{picture}
\caption{Labelled graphs of level one based on motifs $M_1, M_2, M_3, M_4, M_5$.} \label{motifsGraphs}
\end{figure}
For instance, any labeling of the triangle $M_1$ yields the same
graph as in any case each of the nodes has the same neighbors. So
the triangle has six automorphisms.  For the pattern $M_2$, the
corresponding graph shown in Fig.~\ref{motifsGraphs} with the
interchanged labels $a$ and $b$ is the same. However, the graph with
the interchanged $c$ and $d$ is not the same anymore. Of course,
this new graph is isomorphic to the initial one. This is because
there is only one nontrivial automorphism of $M_2$: the one which
interchanges $a$ and $b$, and preserves $c$ and $d$.

Let  ${\sf G}' \subset {\sf G}$ and  ${\sf G}''\subset {\sf G}$ and
there exists an isomorphism $\phi : {\sf G}' \rightarrow {\sf G}''$.
Then we can consider $\phi$ as an equivalence  ${\sf G}' \sim {\sf
G}''$. The equivalence class of ${\sf G}'$ is defined as the set
$[{\sf G}']=\{{\sf G}'' : {\sf G}' \sim {\sf G}'' \}$. It is called
\textit{motif}. If the number of appearances of motif $[{\sf G}']$
in a network $ {\sf G}$ is higher than the number of its appearances
in the Erd\H{o}s-R\'enyi random graph, then $[{\sf G}']$ is called
\textit{network motif}.

Now we present the notion of a \textit{random} graph, which we use
in this work. The {\it random graph model} is defined to be a pair
consisting of an \textit{underlying} graph ${\sf G}=({\sf V}, {\sf
E})$ and a probability space $({\sf E},\mathcal{E}, { P})$. If ${\sf
G}$ is finite, as $\mathcal{E}$ one can take the set of all subsets
of ${\sf E}$. In the sequel, we deal with such random graph models
only. Thus, for ${\sf E}'\in \mathcal{E}$, we say that ${\sf E}'$
has been picked {\it at random} with probability $P({\sf E}')$.  In
many models,  the bonds are being picked independently  with
probability which may depend on the bond. In this case, one deals
with a random graph model with independent bonds. For such graphs,
\begin{equation}
  \label{prob}
  P({\sf E}') = \prod_{e\in {\sf E}'} p(e),
\end{equation}
where $p(e)$ is the probability of picking bond $e$. The set of
graphs $$\{{\sf G}' = ({\sf V}, {\sf E}')\}_{{\sf E}'\in
\mathcal{E}}$$ is called the graph {\it ensemble} -- each ${\sf G}'$
is being picked at random from this ensemble. Now suppose that we
have two random graph models with independent bonds. We have to
specify the definition of isomorphism for such graphs. Let  ${\sf
G}_1=({\sf V}_1, {\sf E}_1)$ and ${\sf G}_2=({\sf V}_2, {\sf E}_2)$
be their underlying graphs and $p_1, p_2$ be their  corresponding
probability (\ref{prob}). Then the map $\phi: {\sf V}_1 \rightarrow
{\sf V}_2$ is said to be the {\it isomorphism of the random graphs}
if there exists isomorphism $f:{\sf G}_1 \rightarrow {\sf G}_2$ (in
the meaning shown previously for non-random graphs) such, that for
every $\langle i, j \rangle \in {\sf E}$ we have
$$ p_1(\langle i,j\rangle)=p_2(\langle f(i),f(j)\rangle). $$

\subsection{The construction}

As was mentioned above, each of our graphs is constructed in an
algorithmic way from the corresponding motif presented in Fig.
\ref{motifs}. As they are random graphs with independent bonds,  we
have to construct the corresponding underlying graphs and to define
the probability of  picking the bonds, cf. (\ref{prob}). In all our
models, the bonds will be of two kinds, which we call \textit{basic
bonds} and \textit{decorations}. Basic bonds are non-random, i.e.,
picked with probability one. Decorating bonds appear with
probability $p\in [0,1]$, which is a parameter of the model. Now we
present the formal construction of the underlying graphs. Let $q$ be
the number of nodes in the corresponding motif, i.e., $q=3$ for
$M_1$ and $q=4$ for the remaining motifs. At step $k=1$, we just
label the nodes of the corresponding motif by $i=1, \dots , q$
and obtain the initial graph $\Lambda_1 = (V_1, E_1)$. All its bonds
are set to be basic. Suppose now that we have $q+1$ copies of
$\Lambda_1$ obtained by the isomorphisms $\phi_2^j$, $j= 0,1, \dots
, q$. Thus, in $j$-th copy the nodes are $\phi^j_2(i)$, $i=1, \dots
, q$. The graph $\Lambda_2$ is obtained from these copies under the
following conditions
\begin{equation}
  \label{gra2}
 \phi_2^0 (i) = \phi_2^i (i), \quad i=1, \dots , q; \qquad \phi_2^i(j) = \phi_2^j (i), \quad i=1, \dots , q, \ i\neq j.
\end{equation}
Thus, the images of $V_1$ under $\phi_2^i$ and $\phi_2^j$ with
$i\neq j$ intersect only at one node where (\ref{gra2}) holds. The
maps $\phi_2^j$, $j= 0,1, \dots , q$ embed $\Lambda_1$ into
$\Lambda_2$. The nodes $\phi_2^i(i)$, $i=1, \dots , q$, are called
the {\it external} nodes of $\Lambda_2$. All other nodes are called
{\it internal}. Thus, $\Lambda_2$ has $q$ external and $q(q-1)/2$
internal nodes. At this stage, we label them by $i= 1 , \dots ,
q(q+1)/2$ in such a way that the external nodes have the same labels
as in $\Lambda_1$, that is,  $\phi_2^i(i) = i$, $i=1, \dots q$. By
construction, the bonds obtained as images under the map $\phi_2^0$
are decorations: they are of the form  $\langle \phi^0_2 (i),
\phi^0_2 (j) \rangle$ where $i$ and $j$ are adjacent in $\Lambda_1$.
From the first condition in (\ref{gra2}) we see that the decorating
bonds connect the external nodes of $\Lambda_2$. The remaining bonds
of $\Lambda_2$ are set to be basic. Now we construct $\Lambda_k$ for
$k \geq 3$ from one copy of $\Lambda_1$ and $q$ copies of
$\Lambda_{k-1}$. Let $\phi_{k}^0$ be the map which produces the copy
of $\Lambda_1$ and $\phi_k^j$, $j=1, \dots ,q$ be the maps which
produce the copies of $\Lambda_{k-1}$. We then impose the conditions
\begin{equation}
  \label{gra3}
  \phi_k^0 (i) = \phi_k^i (i), \quad i=1, \dots , q; \qquad \phi_k^i(j) = \phi_k^j (i), \quad i=1, \dots , q, \ i\neq j
\end{equation}
and obtain $\Lambda_k$. Thus,  $\phi_k^0$ embeds $\Lambda_1
\rightarrow \Lambda_k$, and $\phi_k^i :\Lambda_{k-1} \rightarrow
\Lambda_k$, $i=1, 2,\dots, q$. As above, the nodes $\phi_k^i (i)$
are set to be external, and the remaining nodes are internal. The
images of $V_2$ under $\phi_k^i$ and $\phi_k^j$ with $i\neq j$
intersect only at one node where (\ref{gra3}) holds. Again we label
the nodes of $\Lambda_k$ is such a way that $\phi_k^i(i) = i$, $i=1,
\dots , q$. Now let us establish which bonds of $\Lambda_{k-1}$ are
decorating and which are basic. As above, the bonds connecting the
external nodes are decorating. The images of decorating bonds of
$\Lambda_{k-1}$ are decorating bonds in $\Lambda_k$; the same is
true also for the basic bonds -- the basic bonds of $\Lambda_k$ are
exactly the images of the basic bonds of $\Lambda_{k-1}$. As above,
by $V_k$ and $E_k$ we denote the sets of nodes and bonds of
$\Lambda_k$, respectively. Thus, for $k\geq 2$ we have $E_k = E_k'
\cup E_k''$, where $E_k'$ (respectively, $E_k''$) consists of basic
(respectively, decorating) bonds. All $\Lambda_k$, $k\in
\mathbb{N}$, are considered as subgraphs of an infinite graph
$\Lambda_\infty$, the structure and properties of which are not
important for the study presented in this article.

Note that the construction principle used above  essentially differs
from that used in
\cite{GriffithsKaufman1982,BleherZalys1988,BleherZalys1989,HinczBerker2006}.
Namely, in our case to obtain $\Lambda_k$ one replaces each {\it
node} of the basic pattern by a copy of the graph $\Lambda_{k-1}$.
In the hierarchical lattices, one replaces a {\it bond}. As we shall
see in the sequel, this leads to essentially different properties of
the resulting graphs. Below in Fig. \ref{construction}, we illustrate
the construction described above for the case where the basic
pattern is the motif $M_1$. In this case, the {\it bare} graph
(which occurs for $p=0$) is the approximating graph for the
Sierpi\'nski triangle. The elements of $E'_2$ (middle graph) and of
$E'_3$ (right-hand graph) are depicted as solid lines, whereas the
elements of $E''_2$  and of $E''_3$ appear as dotted lines. We omit
some dotted lines to indicate that they are random and hence may be
absent in a given realization of the graph. Note that $\Lambda_3$
can be viewed as the triangle composed from three copies of
$\Lambda_2$. In Fig. \ref{motif2}, we present the construction of
the bare graph $\Lambda_3$ corresponding to $M_2$. In contrast to
the former case, it is not a planar graph. In Fig. 4, we construct
the bare graph $\Lambda_2$ for motif $M_3$. One observes that in
this picture the node $c$ of the lower left-hand quadrat (i.e.
quadrat $a$) is glued up with node $a$ of the upper right-hand
quadrat. It is interesting that the corresponding fractal can be
obtained by the following procedure, resembling the one which yields
the Sierpi\'nski triangle. One takes the full quadrat and cuts it
out into four equal quadrates,  not cutting the external lines. Then
one glues up the vertices of the smaller quadrates as depicted and
proceeds with cutting out the smaller quadrates. The fractal which
one obtains from $M_5$ is a three dimensional version of the
Sierpi\'nski triangle. One takes the full tetrahedra and cuts out
its inner one fourth in such a way that the remaining four
tetrahedrae are glued up according to the rule: vertex $b$ of
tetrahedra $a$ is glued up with vertex $a$ of tetrahedra $b$, etc.

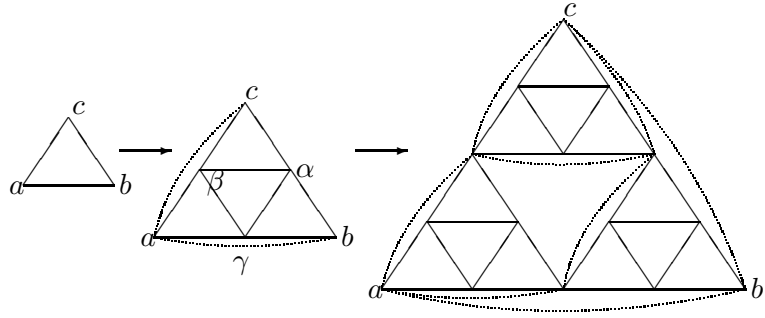
\begin{figure}[htbp]
\centering
\unitlength 0,23mm
\begin{picture}(400, 180)
\trojkatABC(10,70,52)
\put(65,90){\vector(1,0){30}}
\trojkatLevelDwa(85,40,52)
\put(77,35){$a$}
\put(192,35){$b$}
\put(130,22){$\gamma$}
\put(137,120){$c$}
\put(115,67){$\beta$}
\put(166,75){$\alpha$}
\qbezier[50](85,40)(93,80)(135,116)
\qbezier[50](85,40)(135,30)(189,40)
\put(200,90){\vector(1,0){30}}
\trojkatLevelDwa(215,10,52)
\trojkatLevelDwa(319,10,52)
\trojkatLevelDwa(267,88,52)
\put(207,5){$a$}
\put(426,5){$b$}
\put(319,168){$c$}
\qbezier[50](215,10)(223,50)(267,88)
\qbezier[50](215,10)(265,0)(319,10)
\qbezier[50](319,10)(325,50)(371,88)
\qbezier[50](267,88)(273,126)(319,166)
\qbezier[50](371,88)(360,126)(319,166)
\qbezier[50](267,88)(315,76)(371,88)
\qbezier[100](215,10)(315,-15)(423,10)
\qbezier[100](423,10)(395,95)(319,166)
\end{picture}
\caption{Construction of the graph $\Lambda_3$ based on $M_1$ } \label{construction}
\end{figure}

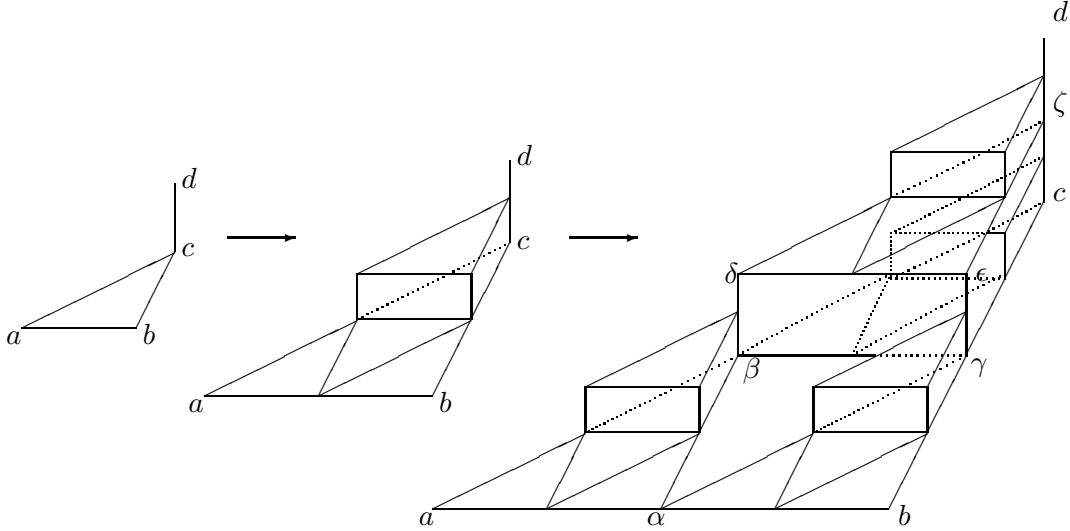
\begin{figure}[ht]
\centering
\begin{picture}(400, 200)
\unitlength 0,3mm
    \put(0,90){\line(1,0){50}}
    \put(0,90){\line(2,1){67}}
    \put(50,90){\line(1,2){17}}
    \put(67,124){\line(0,1){30}}

    \put(-7,83){$a$}
    \put(53,83){$b$}
    \put(70,122){$c$}
    \put(70,152){$d$}

\put(90,130){\vector(1,0){30}}

    \put(80,60){\line(1,0){50}}
    \put(80,60){\line(2,1){67}}
    \put(130,60){\line(1,2){17}}
    \put(147,94){\line(0,1){20}}

    \put(130,60){\line(1,0){50}}
    \put(130,60){\line(2,1){67}}
    \put(180,60){\line(1,2){17}}
    \put(197,94){\line(0,1){20}}

    \put(147,94){\line(1,0){50}}
    \qbezier[30](147,94)(180,110)(214,128)
    \put(197,94){\line(1,2){17}}
    \put(214,128){\line(0,1){20}}

    \put(147,114){\line(1,0){50}}
    \put(147,114){\line(2,1){67}}
    \put(197,114){\line(1,2){17}}
    \put(214,144){\line(0,1){20}}

    \put(73,53){$a$}
    \put(183,53){$b$}
    \put(217,125){$c$}
    \put(217,162){$d$}

\put(240,130){\vector(1,0){30}}

\put(174,3){$a$}

    \put(180,10){\line(1,0){50}}
    \put(180,10){\line(2,1){67}}
    \put(230,10){\line(1,2){17}}
    \put(247,44){\line(0,1){20}}

    \put(230,10){\line(1,0){50}}
    \put(230,10){\line(2,1){67}}
    \put(280,10){\line(1,2){17}}
    \put(297,44){\line(0,1){20}}

    \put(247,44){\line(1,0){50}}
    \qbezier[30](247,44)(270,55)(314,78)
    \put(297,44){\line(1,2){17}}
    \put(314,78){\line(0,1){20}}

    \put(247,64){\line(1,0){50}}
    \put(247,64){\line(2,1){67}}
    \put(297,64){\line(1,2){17}}
    \put(314,94){\line(0,1){20}}


\put(274,3){$\alpha$}
\put(384,3){$b$}
    \put(280,10){\line(1,0){50}}
    \put(280,10){\line(2,1){67}}
    \put(330,10){\line(1,2){17}}
    \put(347,44){\line(0,1){20}}

    \put(330,10){\line(1,0){50}}
    \put(330,10){\line(2,1){67}}
    \put(380,10){\line(1,2){17}}
    \put(397,44){\line(0,1){20}}

    \put(347,44){\line(1,0){50}}
    \qbezier[30](347,44)(370,55)(414,78)
    \put(397,44){\line(1,2){17}}
    \put(414,78){\line(0,1){20}}

    \put(347,64){\line(1,0){50}}
    \put(347,64){\line(2,1){67}}
    \put(397,64){\line(1,2){17}}
    \put(414,94){\line(0,1){20}}

\put(316,68){$\beta$}
\put(416,71){$\gamma$}
\put(452, 146){$c$}
    \put(314,78){\line(1,0){50}}
    \qbezier[30](314,78)(345,96)(381,114)
    \qbezier[20](364,78)(372,96)(381,114)
    \put(381,112){\line(0,1){2}}
    \qbezier[10](381,114)(381,123)(381,132)

    \put(364,78){\line(1,0){10}}
    \qbezier[15](374,78)(394,78)(414,78)
    \qbezier[30](364,78)(395,96)(431,114)
    \put(414,78){\line(1,2){17}}
    \put(431,112){\line(0,1){20}}

    \qbezier[20](381,112)(404,112)(431,112)
    \qbezier[30](381,112)(420,130)(448,146)
    \put(431,112){\line(1,2){17}}
    \put(448,146){\line(0,1){20}}

    \put(423,132){\line(1,0){8}}
    \qbezier[20](381,132)(402,132)(423,132)
    \qbezier[35](381,132)(414,148)(446,165)
    \put(431,132){\line(1,2){17}}
    \put(448,162){\line(0,1){20}}

    \put(308,110){$\delta$}
    \put(418,111){$\epsilon$}
    \put(452,186){$\zeta$}
    \put(452, 226){$d$}
    \put(314,114){\line(1,0){50}}
    \put(314,114){\line(2,1){67}}
    \put(364,114){\line(1,2){17}}
    \put(381,148){\line(0,1){20}}

    \put(364,114){\line(1,0){50}}
    \put(364,114){\line(2,1){67}}
    \put(414,114){\line(1,2){17}}
    \put(431,148){\line(0,1){20}}

    \put(381,148){\line(1,0){50}}
    \qbezier[30](381,148)(420,166)(448,182)
    \put(431,148){\line(1,2){17}}
    \put(448,182){\line(0,1){20}}

    \put(381,168){\line(1,0){50}}
    \put(381,168){\line(2,1){67}}
    \put(431,168){\line(1,2){17}}
    \put(448,198){\line(0,1){20}}

\end{picture}
\caption{Construction of the bare graph $\Lambda_3$ based on  $M_2$} \label{motif2}
\end{figure}

\begin{figure}[htbp] 
\centering
\begin{picture}(300, 120)
\qbezier[150](50,70)(75,80)(100,70)
\qbezier[150](40,30)(65,40)(85,35)
\qbezier[150](50,70)(25,50)(40,30)
\qbezier[150](100,70)(85,50)(85,35)
    \put(43,70){$d$}
    \put(103,70){$c$}
    \put(33,27){$a$}
    \put(88,30){$b$}
\put(110,50){\vector(1,0){30}}
\qbezier[150](165,90)(190,100)(215,90)
\qbezier[150](155,50)(180,60)(200,55)
\qbezier[150](165,90)(150,70)(155,50)
\qbezier[150](215,90)(200,70)(200,55)
\qbezier[150](200,55)(230,60)(255,50)
\qbezier[150](195,10)(220,20)(245,10)
\qbezier[150](200,55)(190,30)(195,10)
\qbezier[150](255,50)(240,30)(245,10)
\qbezier[150](215,90)(240,80)(265,90)
\qbezier[40](205,45)(230,40)(255,50)
\qbezier[100](215,90)(217,70)(212,56)
\qbezier[10](212,56)(210,50)(205,45)
\qbezier[150](265,90)(270,70)(255,50)
\qbezier[100](155,50)(180,40)(196,43)
\qbezier[8](196,43)(200,43)(205,45)
\qbezier[150](145,10)(170,0)(195,10)
\qbezier[150](155,50)(160,30)(145,10)
\qbezier[30](205,45)(210,30)(195,10)
    \put(158,90){$d$}
    \put(268,90){$c$}
    \put(138,7){$a$}
    \put(248,7){$b$}
\end{picture}
\caption{Construction of the bare graph $\Lambda_2$ based on $M_3$ }
\end{figure}
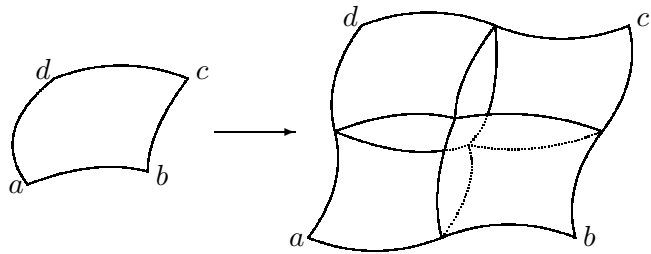

\subsection{Degree distribution}
\label{SS24}

Now we turn to the structural properties of the graphs constructed
above. Here and subsequently, $q$ and $r$ stand for the number of
nodes and bonds in the corresponding motif, respectively. By the
construction described above, the number of nodes in $\Lambda_k$ is
$|V_k|=q|V_{k-1}|-q(q-1)/2$ and $|V_1|=q$. Likewise,
 the expected number of bonds is $|E_1|=r$ and
$\langle|E_k|\rangle=q\langle|E_{k-1}|\rangle+rp$. Hence
\begin{equation}
 \label{VkEk}
|V_k| = \frac{q^{k} +q}{2}, \qquad \langle |E_k|\rangle = r q^{k-1}
+ r p \frac{q^{k-1}-1}{q-1}, \qquad k \in \mathbb{N}.
\end{equation}
As was mentioned above, the degree distribution is an important
characteristic of the graph. In contrast to the Erd\H{o}s-R\'enyi
type graphs, the distribution of the random variable $n(i)$ in our
graphs depends on the type of $i$. Thus, the simplest way to
describe this distribution is to average $n(i)$ over the nodes of a
given $\Lambda_k$, that is, to consider $
  n_k = |V_k|^{-1} \sum_{i \in V_k} n(i)$.
 Let $\langle n_k \rangle$ be the expected value of $n_k$. Then
\begin{equation}
 \label{nk}
\langle n_k \rangle = 2 \frac{\langle|E_k|\rangle}{ |V_k|} = \frac{4 r}{q(q-1)}\Big( q-1+p -  \frac{q-1 + 2p}{ q^{k-1}+1}\Big).
\end{equation}
However, this result gives only partial information about the node
degree distribution. To get more let us analyze the structure of the
node sets $V_k$, $k=1 , 2 , \dots  $. For a given $\Lambda_k$ and $l
= 1 , \dots , k$, let $V^{(l)}_k$ be the set of nodes $i\in V_k$
which have the same degree distribution, independent of $k$ for
$l\leq k-1$. For the graphs based on regular motifs $M_1$, $M_3$ and
$M_5$, $V^{(l)}_k$ consists of the nodes which are external for some
$\Lambda_l$ and, at the same time, are internal for any
$\Lambda_{l+1}$. Here we mean those $\Lambda_l$'s which are
subgraphs for $\Lambda_k$. As an example, let us consider the graph
$\Lambda_2$ based on $M_1$, see the middle graph in Fig.
\ref{construction}. The nodes $a$, $b$, and $c$ constitute
$V^{(2)}_2$, whereas the remaining nodes constitute $V^{(1)}_2$. For
$M_2$ and $M_4$, this partition is more complicated and will be
discussed below. First we analyze  $\Lambda_k$ based on complete
motifs $M_1$ and $M_5$. The elements of $V^{(k-1)}_k$ are exactly
the nodes at which the subgraphs $\Lambda_{k-1}^j$, $j=1, \dots , q$
are glued up to form $\Lambda_k$, whereas the elements of
$V^{(k)}_k$ are exactly the external nodes of $\Lambda_k$. Then
$|V^{(k)}_k|=q$ and $|V^{(k-1)}_k| = q(q-1)/2$. For $l< k-1$, we
have $|V^{(l)}_{k}| = q |V^{(l)}_{k-1}|$, which can be solved to
yield
\begin{equation}
  \label{Vklq}
 |V^{(l)}_k | = \frac{1}{2} q^{k-l}(q-1), \quad l = 1 , \dots , k-1, \qquad |V^{(k)}_k| = q.
\end{equation}
The degrees of $i\in V^{(1)}_k$ are non-random as these nodes
receive no decorating bonds. For such $i$, $n(i) =
\sum_{j}n^{(0)}(j)$, where $n^{(0)}(j)$ is the degree of the
corresponding node in the basic pattern, and the sum is taken over
all such patterns which are glued up. By the symmetry of $M_1$ and
$M_5$, we have that $n(i) = 4$ for $M_1$ and $n(i) = 6$ for $M_5$.
For $i \in V^{(l)}_k$, $l= 2, 3, \dots , k-1$, we have $n(i)=
\tilde{n}(i) + \nu(i)$, where $\tilde{n}(i)$ is non-random and has
to be calculated as just described. The summand $\nu(i)$ is the
number of decorating bonds attached to $i$. For $l = 1 , \dots ,
k-1$ and $i\in V_{k}^{(l)}$, we have $\tilde{n}(i) = 2(q-1)$ and
$\nu (i)$ takes values $ \nu = 0, 1 , 2 , \dots , 2(q-1) (l-1)$,
with probability
\begin{equation}
  \label{Probq}
{\rm Prob} \left(\nu (i \right) = \nu)  = \genfrac{(}{)}{0pt}{}{2(q-1)(l-1)}{ \nu}
 p^\nu (1-p)^{2(q-1)(l-1)-\nu}.
\end{equation}
For $i \in V^{(k)}_k$, $\nu (i)$ takes values $0, 1 , \dots ,
(q-1)(k-1)$. As is usual in the theory of real world networks, which
are in fact non-random, the randomness manifests itself as the
random choice of a node. If we apply this principle here, then
(\ref{Probq}) can be considered as the conditional probability
distribution, conditioned at the event that the node $i$ is been
picked from the set $V^{(l)}_k$. The probability  of this event is
taken to be proportional to the number of elements, that is,
\begin{eqnarray}
  \label{ProbiinVkl}
  {\rm Prob} \left( i \in V^{(l)}_k \right) & = & \frac{|V^{(l)}_k|}{|V_k|} = \frac{q-1}{1 + q^{1-k}} q^{-l}, \quad l \leq k-1, \\[.2cm]
{\rm Prob} \left( i \in V^{(k)}_k \right) & = & \frac{2}{q^{k-1} + 1}. \nonumber
\end{eqnarray}
Now we take the expectation of $n(i)$ with respect to this
distribution and obtain\footnote{Detailed calculations of this and
similar quantities are given in Appendix below.}
\begin{equation}
  \label{nK1}
 \langle n_k \rangle = \frac{q^{k-1}(2q-2+2p)-2p}{q^{k-1}+1},
\end{equation}
which agrees with (\ref{nk}). In the same way we find the second
moment
\[ \langle n^2 \rangle = 4(q-1)^2 + (8q-6)p + (4q+2)p^2.\]
In order to figure out the limit $k\rightarrow +\infty$ of the
distribution given by (\ref{Probq}) and (\ref{ProbiinVkl}) we
calculate its characteristic function, cf. (\ref{nk1}),
\begin{eqnarray*}
\varphi_k (t) & = &  \sum_{l=1}^{k-1} \sum_{\nu =0}^{2(q-1)(l-1)} e^{( 2(q-1) + \nu)\text{i}t} \frac{(q-1)\cdot q^{-l}}{ 1 + q^{1-k}}\cdot\\[.17cm] \nonumber
&&\cdot {
\genfrac{(}{)}{0pt}{}{2(q-1)(l-1)}{\nu}} p^\nu (1-p)^{2(q-1)(l-1)-\nu} + \\[.2cm]
& + & \sum_{\nu=0}^{(q-1) (k-1)} e^{(q-1+ \nu)\text{i}t} \frac{2}{q^{k-1} + 1}  \\[.17cm] \nonumber
&\times &{\genfrac{(}{)}{0pt}{}{(q-1)(k-1)}{\nu}} p^\nu
(1-p)^{(q-1)(k-1)-\nu}.\nonumber
\end{eqnarray*}
Then the limiting characteristic function is
\begin{equation}
  \label{fCharMq}
\varphi (t) = \frac{(q-1)  e^{ 2(q-1) \text{i} t}}{  q -\left(e^{\text{i}t} p + 1 - p \right)^{2(q-1)} }.
\end{equation}
As was mentioned above for the graph based on $M_3$, the same node
partition can be applied also for the graphs based on $M_1$ and
$M_5$. Motif $M_3$ is regular hence (\ref{Vklq}) and
(\ref{ProbiinVkl}) still holds. Here for $l=1,2,\ldots, k-1$ and $i
\in V_k^{(l)}$ we have $n(i)=4+\nu(i)$ and $\nu(i)$ takes values
$\nu=0,1,2,\ldots,4(l-1)$ with probability
\[
\text {Prob} \left(\nu (i \right) = \nu)  = \genfrac{(}{)}{0pt}{}{4(l-1)}{\nu} p^\nu (1-p)^{4(l-1)-\nu}.
\]
The external node $i \in V_k^{(k)}$ has random degree
$\nu(i)=0,1,\ldots,2(k-1)$. Taking the expectation of $n(i)$ with
respect to this distribution one obtains, see Appendix,
\begin{equation}
  \label{nKK}
\langle n_k \rangle =
4+\frac{4}{3}\left(p-\frac{3+2p}{4^{k-1}+1}\right),
\end{equation}
and the $k\to +\infty$ limit of the  second moment
\begin{equation*}
\langle n^2 \rangle = 16+12p+\frac{68}{9}p^2.
\end{equation*}
The characteristic function is
\begin{eqnarray*}
  \varphi_k (t)  &  = & \frac{3 e^{ 4 \text{i} t}}{1 + 4^{1-k}} \cdot \frac{1 - 4^{1-k} \left(e^{\text{i}t} p + 1 - p \right)^{4(k-1)} }{ 4 -\left(e^{\text{i}t} p + 1 - p \right)^4 } + \\ [.2cm]
& + & \frac{2 e^{2 \text{i}t}}{4^{k-1} +1} \left(e^{\text{ i}t} p +
1 - p \right)^{2(k-1)} , \quad {\rm i} = \sqrt{-1}, \nonumber
\end{eqnarray*}
which in the limit $k \rightarrow \infty$ takes the form
\begin{equation} \label{fCharM3}
\varphi (t) = \frac{3  e^{ 4 \text{i} t}}{  4 -\left(e^{\text{i}t} p + 1 - p \right)^4 }.
\end{equation}
It remains to analyze the graphs based on $M_2$ and $M_4$. Label the
nodes of $M_2$ as shown in the Fig. \ref{motifsGraphs}. For the
external nodes of the corresponding graph, we have
\begin{eqnarray} \nonumber
n(a)=n(b) = 2 + \nu(a), && \nu=0,1,\ldots,2(k-1),\\ \nonumber
n(c) = 3 + \nu(v), && \nu=0,1,\ldots,3(l-1),\\ \nonumber
n(d) = 1 + \nu(d), && \nu=0,1,\ldots,k-1.
\end{eqnarray}
For each $l=1,2,\ldots,k-1$ set the $V_k^{(l)}$ consists of three
subsets of the same cardinality with the following degrees
\begin{eqnarray} \nonumber
n(i) = 3 + \nu(i),&&  \nu=0,1,\ldots,3(l-1),\\ \nonumber
n(i) = 4 + \nu(i),&&  \nu=0,1,\ldots,4(l-1),\\ \nonumber
n(i) = 5 + \nu(i),&&  \nu=0,1,\ldots,5(l-1).
\end{eqnarray}
This yields, see Appendix,
\begin{equation}
  \label{nKK1}
 \langle n_k \rangle =
 4+\frac{4}{3}\left(p-\frac{3+2p}{4^{k-1}+1}\right),
\end{equation}
and the second moment
\begin{equation*}
\langle n^2 \rangle = \frac{50}{3}+\frac{112}{9}p+\frac{214}{27}p^2.
\end{equation*}
For the characteristic function, we have
\begin{eqnarray*}
  \varphi_k (t)  & = & \frac{1}{1+4^{1-k}} \Big[
e^{3\text{i} t}\frac{1-4^{1-k}(pe^{\text{i} t}+1-p)^{3(k-1)}}{4-(pe^{\text{i} t}+1-p)^3} + \\ [0.2cm] \nonumber
&  + &\quad e^{4\text{i} t}\frac{1-4^{1-k}(pe^{\text{i} t}+1-p)^{4(k-1)}}{4-(pe^{\text{i} t}+1-p)^4} + \\ [0.2cm]
&  + & \quad e^{5\text{i} t}\frac{1-4^{1-k}(pe^{\text{i} t}+1-p)^{5(k-1)}}{4-(pe^{\text{i} t}+1-p)^5}\Big] +  \nonumber \\[0.2cm]
\nonumber
&+ & \frac{2}{4^{k}+4}\Big[ 2e^{2\text{i} t}(pe^{\text{i} t}+1-p)^{2(k-1)} + \\[0.2cm] \nonumber
& +& \quad e^{3\text{i} t}(pe^{\text{i} t}+1-p)^{3(k-1)} +
e^{\text{i} t}(pe^{\text{i} t}+1-p)^{k-1}\Big], \nonumber
\end{eqnarray*}
which in the limit $k \rightarrow \infty$ yields
\begin{eqnarray}\label{fCharM2}
\varphi (t)&\xrightarrow{} & \frac{  e^{ 3 \text{i} t}}{  4 -\left(e^{\text{i}t} p + 1 - p \right)^3 } + \frac{e^{ 4 \text{i} t}}{  4 -\left(e^{\text{i}t} p + 1 - p \right)^4 } +\\[.2cm] \nonumber
 & &+ \ \frac{  e^{ 5 \text{i} t}}{  4 -\left(e^{\text{i}t} p + 1 - p \right)^5 }.
\end{eqnarray}
For the graphs based on $M_4$, we obtain
\begin{equation}\nonumber
\langle n_k \rangle = 5+\frac{5}{3}\Big(p-\frac{3+2p}{4^{k-1}+1}\Big)
\end{equation}
\begin{equation}\nonumber
\langle n^2 \rangle = \frac{76}{3}+\frac{167}{9}p+\frac{335}{27}p^2,
\end{equation}
\begin{eqnarray}\label{fCharM4}
\varphi(t) & = & \frac{  e^{ 4 \text{i} t}}{2\big(  4 -\left(e^{\text{i}t} p + 1 - p \right)^4 \big)} + \frac{2e^{ 5 \text{i} t}}{  4 -\left(e^{\text{i}t} p + 1 - p \right)^5 } +\\[0.2cm]
 \nonumber
 & + &  \frac{  e^{ 6 \text{i} t}}{ 2\big( 4 -\left(e^{\text{i}t} p + 1 - p \right)^6\big) }.
\end{eqnarray}
For all our graphs, the limiting characteristics functions  can be
continued to functions analytic in some complex neighborhood of the
point $t=0$. This means that the limiting node degree distribution
has all moments and hence cannot be of scale-free type\footnote{For
scale-free graphs, the node degree distribution is  $P(k) = C
k^{-\gamma}$, $k \geq 1$, $\gamma>1$; hence,  $\sum_{k=1}^\infty k^m
P(k)$  diverges for all $m\geq \gamma -1$. }.  Another observation
here is that the characteristic function of the Poisson distribution
\[\varphi_{\rm Poisson} (t) = \exp\left[c \left( e^{{\rm i}t} - 1 \right) \right],\]
can be continued to a function analytic on the whole complex plane.
Therefore, the  degree distributions in our graphs with $p>0$ are
intermediate as compared to the Poisson and scale-free
distributions. For $p=0$, our functions (\ref{fCharMq}),
(\ref{fCharM3}), (\ref{fCharM2}), (\ref{fCharM4}) are also entire.

\subsection{Amenability}

The property of our graphs which we address now is {\it
amenability}. Let $G = (V, E)$ be a countable graph with node set
$V$ and bond set $E$. For a finite $\mathit{\Delta} \subset V$, by
$\partial \mathit{\Delta}$ we denote the set of  nodes which are not
in $\mathit{\Delta}$ but have neighbors in $\mathit{\Delta}$. By
$|\mathit{\Delta}|$ and $|\partial \mathit{\Delta}|$ we denote the
number of elements in these sets. The graph $G$ is said to be {\it
amenable} if there exists a sequence of finite node sets
$\{\mathit{\Delta}_k\}_{k\in \mathbb{N}}$, such that
\begin{equation}
 \label{amenab}
\lim_{k\rightarrow +\infty}  \frac{|\partial \mathit{\Delta}_k|}{|\mathit{\Delta}_k|} = 0.
\end{equation}
If such a limit is positive for any sequence
$\{\mathit{\Delta}_k\}_{k\in \mathbb{N}}$, the graph is called {\it
nonamenabile}. Sometimes, sequences for which (\ref{amenab}) holds
are called {\it Van Howe} sequences. Cayley trees, except for
$\mathbb{Z}$, are nonamenable. Let us  consider the underlying
graphs of our random graphs. Due to their hierarchical structure, it
is convenient to check (\ref{amenab}) for the sequence of node sets
of $\Lambda_k$, that is for $\{V_k\}_{k\in \mathbb{N}}$. By the
construction of $\Lambda_k$, the inner boundary of each $V_k$ is the
set of all its external nodes, the number of which is equal to the
number of nodes in the corresponding motif, i.e. it is $q$. By
construction,  $q-1$ of them become inner nodes of $\Lambda_{k+1}$,
and receive new $k(q-1)$ neighbors (and none in the next steps). The
remaining one becomes an external node of $\Lambda_{k+1}$. We can
choose $\{V_k\}_{k\in \mathbb{N}}$ in the way that this external
node becomes an internal one in the next step. And then it receives
new $(q-1)(k+2)$ neighbors outside $\Lambda_k$ and none in the next
steps. Then for all the graphs we obtain
\begin{equation}\nonumber
\frac{|\partial \Lambda_k|}{|\Lambda_k|} = \frac{k(q-1)^2+(q-1)(k+2)}{\frac{1}{2}(q^k+q)} \xrightarrow{k \rightarrow \infty} 0,
\end{equation}
which means that all our random graphs are amenable with probability
one.

\subsection{Clustering}

For a given node $i\in V$ of degree $n(i)$, let $N(i)$ be the number
of bonds linking its neighbors with each other, which is the number
of triangles with vertex $i$. Clearly, $N(i) \leq n(i) [n(i)-1]/2$
and the maximum value of this parameter is attained for complete
graphs where each node is a neighbor to all other ones. Thus, the
quantity
\begin{equation*}
Q(i) : = \frac{2 N(i)}{n (i) [ n(i) -1]}
\end{equation*}
characterizes clustering at node $i$. Then the clustering of our
graphs we define as
\begin{equation*}
Q = \lim_{k\rightarrow +\infty} \frac{1}{|V_k|} \sum_{i \in V_k} Q(i).
\end{equation*}
Note that for many graphs, e.g., for trees or bipartite graphs, one
has $Q(i) =0$ for any node $i$, see also
\cite{ErdosKleitmanRothschild, FutornyUstimenko}. For random graphs,
the degree $n(i)$, as well as the parameter $N(i)$, are random. The
calculation of $Q$ in this case is much more involved. We will
address it in a forthcoming paper.  Here we only compare the values
of $Q$ obtained for the bare and fully decorated versions of our
graphs, i.e., for $p=0$ and $p=1$.

For the bare graph $\Lambda_k$ based on $M_1$, we have $n(i)=4$ for
internal node $i \in V_k$ and $n(i)=2$ for external node $i \in
V_k$. Besides
\begin{equation}
\nonumber
N(i)=\left \{
    \begin{array}{ll}
       3 & \qquad i \in V_k^{(1)},\\
       2 & \qquad i \in V_k^{(l)},~l=2,3,\ldots, k-1,\\
         1 & \qquad i \in V_k^{(k)},
    \end{array}
\right.
\end{equation}
which follows directly from the construction of the graphs. By
(\ref{VkEk}) and (\ref{Vklq}) one gets
\begin{eqnarray}\nonumber
\frac{1}{|V_k|}\sum_{i \in V_k} Q(v)
 & =&\frac{1}{|V_k|}\sum_{l=1}^k\sum_{i \in V_k^{(l)}} Q(i) =\\[0.2cm]\nonumber
 & =&\frac{|V_k^{(1)}|}{2|V_k|}
+\frac{|V_k|-|V_k^{(1)}|-3}{3|V_k|}+\frac{3}{|V_k|}
=\frac{1}{3}+\frac{|V_k^{(1)}|}{6|V_k|}+\frac{2}{|V_k|}.
\end{eqnarray}
Hence,
\begin{equation}\nonumber
Q = \lim_{k\rightarrow +\infty} \frac{1}{|V_k|} \sum_{i \in V_k} Q(i)=\frac{4}{9}=0,4444\ldots.
\end{equation}
For the fully decorated graph based on $M_1$, internal node $i \in
V_k^{(l)}$, $l=1,2,\ldots,k-1$, has degree $n(i)=4l$ and $N(i)=4l$
whereas for external node $i \in V_k^{(k)}$, we have $n(i)=2k$ and
$N(i)=2k-1$. Then
\begin{eqnarray}\nonumber
\frac{1}{|V_k|}\sum_{i \in V_k} Q(i)
 & =&\frac{1}{|V_k|}\sum_{l=1}^k\sum_{i \in V_k^{(l)}} Q(v)=\\[0.2cm] \nonumber
& =&\frac{1}{|V_k|}\left(\sum_{l=1}^{k-1}\sum_{i \in V_k^{(l)}}\frac{2\cdot 4l}{4l(4l-1)}+|V_k^{(k)}|\frac{2(2k-1)}{2k(2k-1)}\right)=\\[0.2cm]
\nonumber
 & =&\frac{1}{|V_k|}\left( 2\cdot 3^k\sum_{l=1}^{k-1} \frac{3^{-l}}{4l-1}+\frac{3}{k}\right)=\frac{4\cdot 3^k}{3^k+3}\sum_{l=1}^{k-1} \frac{3^{-l}}{4l-1}+\frac{6}{k(3^k+3)}.
\end{eqnarray}
Hence
\begin{eqnarray}\nonumber
 Q = 2\cdot 3^{-1/4} \arctan 3^{-1/4} - 3^{-1/4} \ln \frac{3^{1/4} +1}{3^{1/4} -1} \approx 0,525897.
\end{eqnarray}
For the bare graph based on $M_5$, one obtains for internal node $i
\in V_k^{(l)}, l=1,2,\dots, k-1$: $n(i)=6$,  $N(i)=8$ for $l=1$, and
$n(i)=6$, $N(i)=6$ for $l \ge 2$. For the fully decorated graph
based on $M_5$ we have for all internal nodes $n(i)=6l$, and
$N(i)=8$ for $l=1$, and $N(i)=12l-3$ for $l \ge 2$. Hence, for the
bare graph, we get
\begin{equation}\label{Q}
 Q = \lim_{k\rightarrow +\infty} \left(\frac{2}{5} + \frac{2|V^{(1)}_k|}{15|V_k|} + \frac{12}{5|V_k|} \right) = 0.5,
\end{equation}
and for the fully decorated graph
\begin{equation}
Q \approx 0.554145,  \nonumber
\end{equation}
which  surprisingly is  quite close to the clustering in the bare
version (\ref{Q}).

\subsection{Small-world property}

There exists one more property of real-world networks which
Erd\H{o}s-R\'enyi type graphs do not share, see e.g.
\cite{Newman2003, BarratWeigt}. It is the so called {\it small-world
property}. To formulate it one needs the following notion. A path in
the graph is a sequence of nodes such that every two consecutive
elements are neighbors to each other. The length of the path is the
number of such consecutive pairs, which is equal to the number of
bonds one passes on the way from the origin to the terminus. If
every two nodes can be connected by a path, the graph is said to be
connected. For a given two nodes, $i$ and $j$, the length of the
shortest path  connecting them is said to be the distance
$\rho(i,j)$ between these nodes. Informally speaking, a graph $G =
(V,E)$ has the small-world property (is a small-world graph) if
every two nodes $i,j\in V$ are {`not too far'} from each other. More
precisely this property is formulated as follows. An infinite graph
$G$ has a small-world property if there exists a sequence of its
connected finite subgraphs $\{G_k\}_{k\in \mathbb{N}}$ with the
following property. Let ${\rm diam}(G_k) = \max_{i,j \in V_k} \rho
(i,j)$ be the {\it diameter} of $G_k$, $k\in \mathbb{N}$, and
$\langle n_k \rangle$ be the average value of the node degree in
$G_k$, that is, $\langle n_k \rangle = 2 |E_k| / |V_k|$. Then the
sequence $\{G_k\}_{k\in \mathbb{N}}$, and hence the graph $G$, are
said to have the small-world property if there exists a positive
constant $C$ such that for all $k\in \mathbb{N}$
\begin{equation*}
{\rm diam} (G_k) \leq C \log_{\langle n_k \rangle} |V_k|.
\end{equation*}
In such graphs, the distances between the nodes scale at most
logarithmically with the size of the graph. Let us consider this
characteristic of our graphs without decorations, i.e., for $p=0$.
For a chosen motif, the diameter or $\Lambda_k$ is the maximum
distance between two external nodes
 \[\text{diam}(\Lambda_k)=\max_{i,j} \rho(i,j),   \quad i, j \in V_k^{(k)}, i \neq j.\]
For the complete motifs $M_1$ and $M_5$ there is
$\text{diam}(\Lambda_1)=1$, and for the other motifs
$\text{diam}(\Lambda_1)=2$. By the construction of $\Lambda_k$,
$k=2,3,\ldots$, it is easily seen that the distances between two
external nodes increases two times ay each step. Hence
\begin{equation*}
\begin{array}{lcl}
\text{diam}(\Lambda_k) = 2^{k-1} & & \text{for }M_1\text{ and }M_5,\\
\text{diam}(\Lambda_k) = 2^{k} & & \text{for }M_2,~M_3\text{ and }M_4,
\end{array}
\end{equation*}
that means that the diameters scale exponentially with the size of
the graph.

For $p=1$, the distance between two chosen external nodes in
$\Lambda_k$, $k=1,2,\ldots$, is $1$. Hence, we have to analyze the
distances between other pairs of nodes. Here we present the results
for graphs based on motif $M_1$ only. The distance between an
internal and an external node in $\Lambda_2$ does not exceed $2$.
The distance between two such nodes in $\Lambda_3$ is not greater
than $3$, and $k$ in $\Lambda_k$. Therefore, to estimate the
distance between two internal nodes in $\Lambda_k$,  one has to find
the greatest $l\leq k-1$ such that these nodes belong to different
$\Lambda_l$. Then add the distances between these nodes and the
common external node. Hence
\begin{equation*}
\text{diam}(\Lambda_k) \leq 2(k-1).
\end{equation*}
Thus, neither of our bare graphs  has the small-world property. At
the same time, this property holds for all fully decorated graphs.

\section{Phase Transitions in the Ising Model}

There exists a deep connection between the properties of Gibbs
random fields of the Ising model and the structural properties of
the underlying graphs, see \cite{H}.  In the physical terminology,
each (pure) Gibbs random field corresponds to a state of thermal
equilibrium of the model, see \cite{Georgii} for more details.
Accordingly, the existence of multiple Gibbs random fields
corresponds to the existence of multiple equilibrium states and
hence to phase transitions. For noninteracting spins, the Gibbs
random field is unique. However, if the interaction is strong enough
and if it is effectively propagated by the underlying graph (due to
high 'connectivity'), the Gibbs fields can be multiple.

 The
Ising model on an infinite graph $G = (V,E)$ is defined by assigning
spin variables  $\sigma_i = \pm 1$, $i \in V$. Two spins, $\sigma_i$
and $\sigma_j$, interact whenever $i\sim j$. The space of spin
configurations is then $\Sigma := \{-1,1\}^V$. It is equipped with
the discrete topology and the corresponding Borel $\sigma$-field. A
Gibbs random field is defined as a probability measure on $\Sigma$
which satisfies a certain condition formulated by means the so
called Gibbs specification, see \cite{Georgii}. The specification in
turn is constructed by means of conditional model Hamiltonians,
defined as follows. For a finite $\varDelta \subset V$, $\xi \in
\Sigma$, and a fixed value of the inverse temperature $\beta >0$,
the conditional Hamiltonian in $\varDelta$ is given by the following
expression
\begin{eqnarray}
  \label{BQ}
- \beta \mathcal{H}_\varDelta (\sigma_\varDelta|\xi) & = & h \sum_{i
\in \varDelta} \sigma_i + \sum_{\{i,j\}\in E_\varDelta}J_{ij}
\sigma_i \sigma_j + \sum_{i \in  \varDelta} \sum_{j \in \varDelta^c
: i \sim j}J_{ij} \sigma_i \xi_j,
\end{eqnarray}
where $\sigma_\varDelta = \{\sigma_i: i \in \varDelta\}$, $h$ is an
external field, and $J_{ij}\in \mathbb{R}$ is the spin-spin
interaction intensity. Note that the latter parameters include
$\beta$. For hierarchical graphs constructed in an algorithmic way,
the infinite graph $(V,E)$ is obtained as limiting object, defined
by means of a system of embeddings which map each finite fragment
into the graph, see \cite{BleherZalys1988,BleherZalys1989}. In this
article, however, we do not follow this way and consider the
annealed case, see \cite{Bovier2006}, where one deals exclusively
with states on such finite fragments. Here we mean the randomness
related to the decorating bonds.

In view of the hierarchical structure of our graphs, we take
$\varDelta= \varDelta_k$ in (\ref{BQ}) to be the set of the inner
nodes of a given $\Lambda_k$ corresponding to motif $M_1$. Then $j$
in the last term in (\ref{BQ}) runs through the set of external
vertices of $\Lambda_k$. We also restrict our consideration to the
case of $h=0$. For $k=1$ we have no internal nodes and no randomness
either. Thus the corresponding Hamiltonian is
\begin{equation}
  \label{A0}\nonumber
-\beta  \mathcal{H}_1 = K (ab + ac + bc),
\end{equation}
here we use the shorthand like $a = \sigma_a$, and $K$ stands for
the interaction intensity corresponding to nonrandom bonds. Recall
that, for $k>1$, by $E_k'$ (resp. $E_k''$) we denote the set of
solid (resp. decorating) bonds of $\Lambda_k$. To take the latter
randomness into account we introduce independent random variables
$\omega \in \{0,1\}^{E_k''}$ such that ${\rm Prob} (\omega_{ij} = 1)
= p$. Then we set $J_{ij}= J^\omega_{ij} = K$ for $\langle
i,j\rangle \in E_k'$, and $J_{ij}=J^\omega_{ij}= L \omega_{ij}$ for
$\langle i,j\rangle \in E_k''$. In general, we assume that $K\neq L$
as the random and nonrandom bonds play different roles in our
constructions. The nonrandom bonds form a skeleton of the graphs,
whereas the random ones increase its connectivity. Moreover, by
setting $K=0$ we can pass to the model defined on a purely random
graph.

 Then, for $k\geq 2$,  (\ref{BQ}) takes the form
\begin{eqnarray}
  \label{A1}
& & -\beta  \mathcal{H}_k (\sigma_{\varDelta_k}|a,b,c) :=  -\beta
\mathcal{H}_{\varDelta_k} (\sigma_{\varDelta_k}|a,b,c)  =
L\left(\omega_{ab} a b + \omega_{ac} a c + \omega_{bc} b c \right)
\\[.2cm] & & + a \sum_{i \in N_k^a} J^\omega_{ai} \sigma_i +
b \sum_{i \in N_k^b} J^\omega_{bi} \sigma_i  +  c \sum_{i \in N_k^c}
J^\omega_{ci} \sigma_i + \sum_{\langle i,j\rangle\in E_k^{in}}
J^\omega_{ij} \sigma_i \sigma_j, \nonumber
\end{eqnarray}
where $N_k^v = \{ i \in V_k: i \sim v\}$ is the set of the neighbors
of $v$ in $\Lambda_k$ and $E_{k}^{in}\subset E_k$ is the set of
bonds connecting the inner nodes of $\Lambda_k$ to each other.

\begin{figure}[htbp]
\centering
\unitlength 0,4mm
\begin{picture}(140, 100)
\trojkatLevelDwa(25,15,52)
\put(17,10){$a$}
\put(132,10){$b$}
\put(70,0){$\gamma$}
\put(77,95){$c$}
\put(40,50){$\beta$}
\put(108,50){$\alpha$}
\put(45,30){$\Lambda_{k-1}^a$}
\put(94,30){$\Lambda_{k-1}^b$}
\put(70,65){$\Lambda_{k-1}^c$}
\qbezier[50](25,15)(33,55)(76,93)
\qbezier[50](25,15)(75,5)(129,15)
\qbezier[50](129,15)(117,55)(76,93)
\end{picture}
\caption{Graph $\Lambda_k$} \label{grafL2doHamiltonianu}
\end{figure}
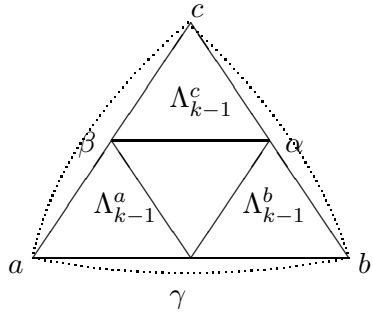
The Hamiltonian in (\ref{A1}) can be rewritten in a recursion way,
see Fig. \ref{grafL2doHamiltonianu}, as follows
\begin{equation}  \label{A2}
\left.
  \begin{array}{rcl}
-\beta \mathcal{H}_k (\sigma_{\varDelta_ k}| a,b,c) & = &
L(\omega_{ab} a b + \omega_{ac} a c + \omega_{bc} b c) -
\beta\mathcal{H}_{k-1} (\sigma_{\varDelta_{k-1}^a}| a,\gamma,\beta) \\[.2cm]& - &
\beta \mathcal{H}_{k-1}(\sigma_{\varDelta_{k-1}^b}| \gamma,b,\alpha)
- \beta \mathcal{H}_{k-1} (\sigma_{\varDelta_{k-1}^c}| \beta
,\alpha,c),
  \end{array}
\right.
\end{equation}
where $\varDelta_{k-1}^x$ stands for the set of inner nodes of
$\Lambda_{k-1}^x$, $x\in \{a,b,c\}$. Then the annealed Gibbs
distribution in $\Lambda_k$ is defined as follows
\begin{eqnarray*}
\pi^\omega_k(\sigma_{\varDelta_ k}|a,b,c) & = & \frac{1}{ Z_k(a,b,c)
} \exp\left( -\beta
 \mathcal{H}_k (\sigma_{\varDelta_ k}| a,b,c) \right), \quad k\geq 2,
\end{eqnarray*}
and $\pi^\omega_1=1$ as $\varDelta_ 1= \emptyset$. Here the
partition function has the form
\begin{eqnarray}\label{A3a}\nonumber
Z_k(a,b,c) & = & \bigg{\langle} \sum_{\sigma^k} \exp\left( - \beta
\mathcal{H}_k (\sigma_{\varDelta_ k}| a,b,c) \right)\bigg{\rangle},
\end{eqnarray}
\begin{equation}
  \label{A5}\nonumber
Z_1 (a,b,c) = \exp \left( K(ab+ac+bc) \right),
\end{equation}
and $\langle \cdot \rangle$ denotes the expectation in $\omega$.

For $k\geq 2$ let $f: \{-1,1\}^{\varDelta_k} \rightarrow \mathbb{R}$
be a local observable, which is a function dependent on
$\sigma_{\Lambda_m}$ with some  $m<k$ such that $\Lambda_m \subset
\Lambda_k$. Set
\begin{equation}  \label{A6}
F_k(f|a,b,c) = \sum_{\sigma_{\varDelta_k}} \bigg{\langle}
f(\sigma_{\varDelta_k}) \pi^\omega_k(\sigma_{\varDelta_k}|a,b,c)
\bigg{\rangle}, \quad k> m.
\end{equation}
The sequence $\{F_k(f|a,b,c)\}_{k\geq m}$ is bounded and thus has
accumulation points. Our aim is to study their dependence on the
values of the boundary spins $a,b,c$.

In view  of the independence of the bond variables $\omega$, we have
\begin{eqnarray}  \label{A8}
Z_k(a,b,c) & = & \bigg{\langle}\exp\left(  L\omega_{ab} a b \right)
\bigg{\rangle}\bigg{\langle}
\exp\left( L\omega_{ac} a c \right) \bigg{\rangle} \bigg{\langle} \exp\left( L\omega_{bc} b c \right)\bigg{\rangle} \\[.2cm]
& \times & \sum_{\alpha, \beta , \gamma}Z_{k-1} (a , \gamma, \beta) Z_{k-1} (\gamma, b, \alpha)  Z_{k-1} (\beta, \alpha, c). \nonumber
\end{eqnarray}
Assume now that the observable $f$ depends on the spins indexed by
$\Lambda_m\subset \varDelta_{k-1}^a$. Then  by (\ref{A2}) and
(\ref{A6}) we have
\begin{eqnarray*}
F_k(f|a,b,c) & = & \frac{1}{Z_k (a,b,c)}   \bigg{\langle}\exp\left( L\omega_{ab} a b \right)
\bigg{\rangle}\bigg{\langle} \exp\left( L\omega_{ac} a c \right) \bigg{\rangle} \bigg{\langle}
\exp\left( L\omega_{bc} b c \right)\bigg{\rangle} \times \\[.2cm] &&
\displaystyle \sum_{\alpha, \beta , \gamma} \bigg[ \bigg{\langle}
\sum_{\sigma_{\varDelta_{k-1}^a} } f(\sigma_{\varDelta_{k-1}^a})
\exp\left( -\beta\mathcal{H}_{k-1} (\sigma_{\varDelta_{k-1}^a}| a,\gamma,\beta) \right)\bigg{\rangle} \times \\[.2cm]
&  &\times  \bigg{\langle}\sum_{\sigma_{\varDelta_{k-1}^b}
}\exp\left(-\beta \mathcal{H}_{k-1}
 (\sigma_{\varDelta_{k-1}^b}| \gamma,b,\alpha) \right)\bigg{\rangle} \times\\[.2cm]
&  &\times \bigg{\langle}\sum_{\sigma_{\varDelta_{k-1}^c}
}\exp\left(-\beta \mathcal{H}_{k-1} (\sigma_{\varDelta_{k-1}^c}|
\beta,\alpha,c) \right)\bigg{\rangle}\bigg], \quad k>m,
\end{eqnarray*}
that can be rewritten as follows
\begin{eqnarray}  \label{A7}
& & F_k(f|a,b,c) = \\[.3cm] & & \quad = \frac{\sum_{\alpha, \beta , \gamma}
F_{k-1}(f|a , \gamma, \beta)Z_{k-1} (a , \gamma, \beta) Z_{k-1}(\gamma, b, \alpha)
 Z_{k-1} (\beta, \alpha, c)}{\sum_{\alpha, \beta , \gamma}Z_{k-1} (a , \gamma, \beta) Z_{k-1} (\gamma, b, \alpha)  Z_{k-1} (\beta, \alpha, c)}.
\nonumber
\end{eqnarray}
From now on we assume that the locality of  $f$ is such that it
$m=1$, see (\ref{A6}), and that the corresponding $\Lambda_1$ is a
subset of $\Lambda_{k-1}^a$ for all $k\geq 2$. Then, in addition to
(\ref{A6}), we set
\begin{equation}  \label{A8a}
F_1(f|a,b,c) = f(a,b,c)= f(a,c, b) >0,
\end{equation}
where we assume also that $f$ is positive and symmetric with respect
to $b \leftrightarrow c$. Next we introduce the following variables
\begin{eqnarray}
  \label{A9}
A_k & := & Z_k (1,1,1) = Z_k (-1,-1,-1), \\[.2cm] B_k & := & Z_k (1,\pm 1,\mp1) = Z_k (-1,\pm 1,\mp1) , \nonumber \\[.2cm]
Y^1_k & := & F_k(f|1,1,1) = F_k(f|-1,-1,-1), \nonumber    \\[.2cm]
Y^2_k & := & F_k(f|1,1,-1) = F_k(f|1,-1,1),     \nonumber \\[.2cm]
Y^3_k &:= & F_k(f|1,-1,-1)=F_k(f|-1,1,1),     \nonumber
\end{eqnarray}
and
\begin{eqnarray}
  \label{A10}\nonumber
\left. \begin{array}{rcl}
R^{+}  & = & \left( p e^L + 1-p\right)^3,\\[.1cm]
R^{-}  & = & \left( p e^L + 1-p\right)\left(  p e^{-L} + 1-p\right)^2 .
\end{array} \right.
\end{eqnarray}
Then from (\ref{A8}) and (\ref{A7}) we get
\begin{eqnarray}
  \label{A11}
\left. \begin{array}{rcl}
A_{k+1}& = & R^+ \left(A_k^3 + 3A_k B_k^2 + 4 B_k^3 \right), \\[.2cm]
B_{k+1}& = &    R^- \left(A_k^2 + 4 A_k B_k^2 + 3B_k^3 \right),
\end{array} \right.
\end{eqnarray}
and
\begin{eqnarray}  \label{A12}
Y^1_{k+1} & = & \frac{Y^1_k A_k(A_k^2 + B_k^2) + 2Y^2_k B_k^2(A_k + B_k) + 2 Y^3_kB_k^3 }{A_k^3 + 3A_k B_k^2 + 4 B_k^3} , \\[.2cm]\nonumber
Y^2_{k+1} & = & \frac{Y^1_k A_k B_k (A_k + B_k) + 2 Y^2_k B^2_k (A_k + B_k) + Y^3_k B_k^2(A_k+ B_k)  }{B_kA_k^2 + 4 A_k B_k^2 + 3B_k^3},  \\[.2cm]
Y^3_{k+1} & = & \frac{2 Y^1_k A_k B_k^2 + 2 Y^2_k B_k^2(A_k+B_k) +
Y^3_k B_k(A_k^2+B_k^2)}{B_kA_k^2 + 4 A_k B_k^2 + 3B_k^3}, \nonumber
\end{eqnarray}
with the initial conditions
\begin{equation}   \label{A13}
A_1 = e^{3K}, \quad B_1 = e^{-K},
\end{equation}
and
\begin{equation*}  
Y^1_1 = f(1,1,1), \quad Y^2_1 = f(1,1,-1), \quad Y^3_1 = f(1,-1,-1).
\end{equation*}
Then for
\begin{equation}  \label{A15}
x_k = A_k/B_k, \quad x_1 = e^{4K} >0,
\end{equation}
by (\ref{A11}) we get
\begin{equation*}
  x_{k+1} = t \phi (x_k),
\end{equation*}
where
\begin{equation}
  \label{A18}
\phi (x) =  \frac{x^3 + 3x + 4}{x^2 + 4x + 3}= \frac{x^2 - x + 4}{x+3},
\end{equation}
and
\begin{equation}
  \label{A16}
t = \frac{R^+}{R^-} = \left( \frac{pe^L + 1-p}{pe^{-L} + 1-p}\right)^2.
\end{equation}
In these notations, (\ref{A12}) can be rewritten in the following
form
\begin{equation}
  \label{a}
Y_{k+1} = T(x_k) Y_k, \quad k\in \mathbb{N},
\end{equation}
where $Y_k$ is the column vector transposed to $(Y^1_k, Y^2_k,
Y^3_k)$ and
\begin{equation}
  \label{a1}
  T(x) = \left( \begin{array}{ll} \frac{x(x^2+1)}{x^3 + 3x + 4} \ &\frac{2(x+1)}{x^3 + 3x + 4} \ \ \ \frac{2}{x^3 + 3x + 4} \\[.2cm]
\frac{x(x+1)}{x^2+4x+3} \ &\frac{2(x+1)}{x^2+4x+3} \ \ \
\frac{x+1}{x^2+4x+3} \\[,2cm]
\frac{2x}{x^2+4x+3} \ &\frac{2(x+1)}{x^2+4x+3} \ \ \
\frac{x^2+1}{x^2+4x+3}
   \end{array} \right).
\end{equation}
Observe that, for each $x>0$, $T(x)$ is a stochastic matrix, which
means that each of its rows consists of nonnegative elements and
sums up to one. Then, for each $k\in \mathbb{N}$, the matrix
\begin{equation}
  \label{a2}
  S_k = T(x_k) T(x_{k-1}) \cdots T(x_2) T(x_1)
\end{equation}
is also stochastic, and the solution of the recursion in (\ref{a})
is
\begin{equation}
  \label{a3}
  Y_{k+1} = S_k Y_1.
\end{equation}
Products of stochastic matrices as in (\ref{a2}) appear in the
theory of inhomogeneous Markov chains, see e.g.
\cite{S,Bremaud,Hajnal}. They also are being used in communication
networks, control theory, parallel computing, and decision making,
see  \cite{Touri,Blond,Gau,Touri1} and the references therein.

As mentioned above, our aim is to study the limits of the sequences
$\{Y^i_k\}$, $i=1,2,3$,  defined in (\ref{A9}) -- (\ref{A13}), and
hence described by (\ref{a}), (\ref{a3}). If, for an arbitrary
$Y_1$, the limit  $Y_{\infty}= \lim_{k\to \infty} Y_k$ is a vector
with all components equal to each other, then the limiting average
(\ref{A6}) is independent of the boundary spins, which corresponds
to the uniqueness of the limiting Gibbs state, and hence of the
state of thermal equilibrium of the model.  In the terminology of
Markov chains, this is related to the {\it ergodicity} of the
sequence $\{T(x_k)\}$. By definition, see e.g. \cite[Definition 1,
page 1479]{Touri1}, such a sequence is ergodic if the product
sequence $\{S_k\}$ as in (\ref{a2}) converges to a stochastic matrix
with identical rows. In \cite{Blond}, such a sequence is called
consensus. In this case, the sequence $\{Y_k\}$ converges to a
vector with identical entries. Likewise, the existence of
subsequences $\{Y_{k_l}\}_{l\in \mathbb{N}}$ convergent to vectors
with nonequal components corresponds to the multiplicity of such
states, and hence of a phase transition. Since we do not introduce
the Gibbs states of our model explicitly, we use the following
\begin{definition}
  \label{1df}
For fixed $K$, $L$ i $p$, the Ising model on our graph is said to be
in an unordered state if for each observable $f$ satisfying
(\ref{A8a}), there exists  $\lim_{k\rightarrow + \infty}
F_k(f|a,b,c)$ independent of $b$ and $c$, and hence of $a$.
Otherwise, the model is said to be in an ordered state.
\end{definition}
As might be seen from (\ref{a}), the limiting properties of the
sequence $\{x_k\}_{k\in \mathbb{N}}$ are crucial for the
corresponding properties of $\{Y_k\}_{k\in \mathbb{N}}$.
\begin{lemma}
  \label{1lm}
 Assume that the
sequence $\{x_k\}_{k\in \mathbb{N}}$ defined  in (\ref{A15}) --
(\ref{A16}) converges to a certain $x_* >0$. Then the sequence
$\{T(x_k)\}_{k\in \mathbb{N}}$ defined in (\ref{a1}) is ergodic.
\end{lemma}
\begin{proof}
 By definition, each row $\tau_i = (\tau_{i1}, \dots, \tau_{in})$
 of a stochastic matrix $T = (\tau_{ij})_{n\times n}$ is a probability distribution. For two
 such rows, we define
 \begin{equation*}
\|\tau_i - \tau_j\| = \frac{1}{2} \sum_{l=1}^n |\tau_{il} -
\tau_{jl}|.
 \end{equation*}
Then the {\it Dobrushin ergodicity coefficient} of $T$ is
\begin{equation*}
  D(T) := \max_{i,j=1, \dots , n} \|\tau_i - \tau_j\|.
\end{equation*}
It can also be written in the form:
\begin{equation*}
  D(T) =1 - \min_{i<j} \sum_{l=1}^n \min\{\tau_{il}; \tau_{jl}\},
\end{equation*}
which yields that $D(T)< 1$ whenever all $\tau_{ij}$ are strictly
positive, see \cite{Gau} for more detail on this issue. On the other
hand, for two stochastic matrices $T$ and $Q$, it is known that, see
\cite[Chapter 3]{Bremaud},
\begin{equation}
  \label{a10}
  D(TQ) \leq D(T) (Q).
\end{equation}
Since the matrix elements of $T(x)$ in (\ref{a1}) are continuous in
$x>0$, one has $T(x_k) \to T(x_*)$, component-wise, as $k\to
+\infty$. Therefore, each element of the latter matrix is strictly
positive, which yields $D(T(x_*))=: \delta <1$. By the mentioned
continuity we also have that, for a given $\epsilon>0$ such that
$\delta + \epsilon <1$, there exists $k_\epsilon$ such that
$D(T(x_k)) < \delta + \epsilon<1$ for all $k>k_\epsilon$. This
yields by (\ref{a2}) and (\ref{a10}) that $D(S_k) \to 0$. As a
sequence of stochastic matrices, $\{S_k\}_{k\in \mathbb{N}}$
contains convergent subsequences, each of which converges to a
stochastic matrix $S$ with strictly positive elements, for which
$D(S)=0$. By the Perron-Frobenius theorem there exists only one such
limit.
\end{proof}
For a given $k\in \mathbb{N}$, let $d(Y_k)$ be the {\it diameter} of
$Y_k$, that is,
\begin{equation*}
 d(Y_k) := \max_{i,j=1,2,3} (Y_k^i - Y_k^j).
\end{equation*}
It is known that, see \cite{Gau}, for each $Y_1$,
\begin{equation*}
 d(S_k Y_1) \leq D(S_k) d(Y_1),
\end{equation*}
which yields $d(Y_k) \to 0$ as $k\to \infty$ if $\{T(x_k)\}_{k\in
\mathbb{N}}$ is ergodic. On the other hand, for each $\epsilon >0$,
one finds $Y_1$ with positive entries such that
\begin{equation}
  \label{a3c}
d(S_k Y_1) > [D(S_k) - \epsilon] d(Y_1).
\end{equation}
As possible limits of $\{x_k\}_{k\in \mathbb{N}}$, there can appear
the solutions of the following equation:
\begin{equation}
  \label{a4}
  x = t \phi(x), \quad x>0.
\end{equation}
First, let us consider the case of $L<0$ where the interaction along
the random bonds is antiferromagnetic. Then $t<1$, see (\ref{A16}),
and the only solution of (\ref{a4}) is
\begin{equation}
  \label{XY13}
 x_* = \frac{ - 3 - t + \sqrt{ 9 + 22 t - 15 t^2}}{ 2 (1-t)} < 1,
\end{equation}
which is clearly positive.  By (\ref{A18}) we have
\begin{equation*}
\phi' (x) = 1 - \left(\frac{4}{x+3} \right)^2.
\end{equation*}
Therefore, the solution (\ref{XY13}) is stable and $x_k \to x_*$ for
all $x_1>0$. Thus, by Lemma \ref{1lm} the model is in an unordered
state for such $L$ and all $K\in \mathbb{R}$ and $p\in [0,1]$. This
possibly holds due to the frustration caused by the motif the graph
is based on. For $L=0$, i.e., for the graph without decorations, we
have $t=1$ and the only solution of (\ref{a4}) is $x_*=1$. In this
case, the model is in an unordered state for all $K\in \mathbb{R}$.

Let us turn now to the case of $L>0$, in which $t>1$. Now (\ref{a4})
has two solutions
\begin{equation*}
    x^{(1)}_* = \frac{3+t-\sqrt{9+22t-15t^2}}{2(t-1)},
    \qquad x^{(2)}_* = \frac{3+t+\sqrt{9+22t-15t^2}}{2(t-1)},
\end{equation*}
which exist and are distinct provided
\begin{equation}
  \label{t}
  t \in (1, 9/5).
\end{equation}
By direct calculations we get that $t \phi'(x^{(1)}_*)< 1$ and $
t\phi_t'(x^{(2)}_*)> 1$, see Fig. \ref{wykres}. Hence, $x^{(1)}_*$
is stable, whereas  $x^{(2)}_*$ is unstable. This means that
\begin{equation}
  \label{tt}
\lim_{k \rightarrow +\infty} x_k = \left\{ \begin{array}{ll} x^{(1)}_* \quad &{\rm if} \ \ x_1 < x^{(2)}_*; \\[.2cm]
 x^{(2)}_* \quad &{\rm if} \ \ x_1 = x^{(2)}_*; \\[.2cm] +\infty \quad &{\rm if} \ \ x_1 > x^{(2)}_* .\end{array} \right.
\end{equation}
Note also that $x^{(1)}_* \rightarrow 1$ and $x^{(2)}_* \rightarrow
\infty$ as $t \rightarrow 1$.
\begin{figure}[htbp]
\unitlength 0,24mm
\centering
\begin{picture}(150, 180)
    \put(0,20){\vector(1,0){180}}
    \put(20,0){\vector(0,1){190}}
    \put(20,20){\line(1,1){160}}
   \qbezier(20,40)(65,13)(180,190)
    \put(2,36){$\frac{4}{3}t$}
    \put(42,25){$x_*^{(1)}$}
    \put(162,148){$x_*^{(2)}$}
    \put(117,162){$t \phi(x)$}
\end{picture}
\caption{Graphical solution of (\ref{a4})} \label{wykres}
\end{figure}
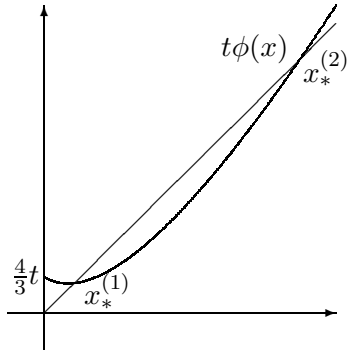
\begin{lemma}
  \label{2lm}
Assume that $x_1>x_*^{(2)}$, and hence $x_k \to \infty$ as $k\to
\infty$. Then the sequence $\{T(x_k)\}_{k\in \mathbb{N}}$ defined in
(\ref{a1}) is not ergodic.
\end{lemma}
\begin{proof}
The proof will be done by showing that the sequence of
\begin{equation*}
  \gamma_k = T_{23}(x_k) + T_{32}(x_k), \quad k\in \mathbb{N},
\end{equation*}
is summable. This will imply that the sequence $\{T(x_k)\}_{k\in
\mathbb{N}}$ fails to have the infinite flow property in the sense
of \cite[Definition 2, page 1479]{Touri1}. This implies in turn, see
\cite[Theorem 1, page 1480]{Touri1} and also \cite{Touri}, the
property in question. Standard linearization yields
\[
x_{k+1} - x_*^{(2)} \geq \varkappa (x_{k} - x_*^{(2)} ), \quad
\varkappa := t \phi'(x_*^{(2)}),
\]
which can be iterated to give
\begin{equation*}
x_{k+1} - x_*^{(2)} \geq \varkappa^k (x_{1} - x_*^{(2)} ).
\end{equation*}
At the same time, by (\ref{a1}) it follows that
\[
\gamma_{k} = \frac{3}{x_k + 3} < \frac{3}{x_k - x_*^{(2)}} \leq
\frac{3 \varkappa^{-(k-1)}}{x_k - x_*^{(2)}} ,
\]
which is summable as $\varkappa >1$.
\end{proof}
By (\ref{a3c}) we obtain from the latter the following
\begin{corollary}
  \label{1rk}
For each $x_1 > x_*^{(2)}$,  there exist $\epsilon >0$ and the
observable $f$ such that, for each $k\in \mathbb{N}$, one finds
$l>k$ with the property $d(Y_l) > \epsilon$. Hence, the model is in
an ordered state.
\end{corollary}
The unstable solution $x_*^{(2)}$ corresponds to the {\it critical
point} which separates two basins of attraction, see (\ref{tt}).
Note that the sequence $\{T(x_k)\}_{k\in \mathbb{N}}$ is still
ergodic since all the entries of $ T(x_*^{(2)})=T(x_k)$, $k\in
\mathbb{N}$, are strictly positive.

Now let us turn to the condition (\ref{t}). It yields
\begin{equation*} 
p < \psi (L) := \frac{3-\sqrt{5}}{\sqrt{5}\exp(L) - 3\exp(-L) +
3-\sqrt{5}}.
\end{equation*}
As $\psi (L)$ is a decreasing function,  the equation $\psi (L) = 1$
has a unique solution
\begin{equation} \label{m20}
L_* = \frac{1}{4} \ln \frac{9}{5}.
\end{equation}
For $L< L_*$, one has $\psi (L) > 1$, which means that  $t< 9/5$ for
all $p \in (0,1]$. For such $L$ and $p$, let $K_*(L,p)$ be the
solution of the equation
\begin{equation*}
\exp(4 K) =   x^{(2)}_* >1.
\end{equation*}
Then (\ref{tt}) can be rewritten in the form
\begin{equation}
  \label{m22}
 \lim_{k \rightarrow +\infty} x_k = \left\{ \begin{array}{ll} x^{(1)}_* \quad &{\rm if} \ \ K < K_* (L,p); \\[.2cm]
 x^{(2)}_* \quad &{\rm if} \ \ K = K_* (L,p); \\[.2cm] +\infty \quad &{\rm if} \ \ K > 4K_* (L,p).\end{array} \right.
\end{equation}
For $t = 9/5$, we have $x^{(1)}_* = x^{(2)}_* =3$, which corresponds
to $K_* = (\ln 3)/4$. In this case, $x_k \rightarrow 3$ if $K \leq
K_*$, and $x_k \rightarrow +\infty$ if $K > K_*$.

For $L > L_*$ there exists $p_*=\psi(L) < 1$ such, that for $p \in
(0, p_*)$, there exists $K_*(L,p)$ with the properties as in
(\ref{m22}). For $p\in [p_*, 1]$, the whole graph of $t\phi$ lies
above the line $t\phi(x) =x$, which means that $x_k \rightarrow
+\infty$ for all initial $x_1\geq 0$.

The results of the analysis just performed can be summarized in the
following form.
\begin{theorem}
  \label{1tm}
The Ising model  on the graph based on $M_1$ and described by the
Hamiltonian (\ref{A1}) has the following properties related to
Definition \ref{1df}:
\begin{itemize}
  \item[(i)] for $L\leq 0$, it is in an unordered state for all
  values of $K\in \mathbb{R}$ and $p\in [0,1]$;
\item[(ii)] for $L\in (0,L_*]$ as in (\ref{m20}) and $p\in (0,1]$,
there exists $K_*(L,p)>0$ such that the model is in an unordered
state for $K\leq K_*(L,p)$, and in an ordered state for $K>
K_*(L,p)$; for $K= K_*(L,p)$, the model is in the critical state;
\item[(iii)] for $L>L_*$, there exists $p_*\in (0,1)$ such that, for
$ p< p_*$, there exists $K_*(L,p)$ with the properties as in item
(ii); for $p\in [p_*,1]$, the model is in an ordered state for all
$K$.
\end{itemize}
\end{theorem}

\section{Concluding Remarks}

In this article, we introduce hierarchical random graphs based on
motifs presented in Fig. \ref{motifs}. The construction principles
resemble those used in \cite{HinczBerker2006}: a nonrandom skeleton
(hierarchical diamond lattice in \cite{HinczBerker2006}) is
accompanied by random bonds. In our case they repeat the motif used
in the construction. As a result, the motif appears at each
hierarchical level. The construction is performed in a rigorous way
and is illustrated by an informal description. The analysis of the
node degree distribution in the constructed graphs is based on
characteristic functions obtained in an explicit form. For $p>0$,
these functions are meromorphic for all motifs. This means that, for
all our graphs, the node degree as a random variable has all moments
with the property $\langle n^m \rangle \sim C^m m!$. Thus, the
degree distributions are intermediate as compared to the Poisson and
scale-free cases. Such properties as clustering and small world
property are studied only for $p=0,1$. In particular, it turns out
that for all motifs the small world property is absent for $p=0$ and
present for $p=1$. Thus, it would be interesting to find out how and
at which value of $p$ is emerges. In Theorem \ref{1tm}, we analyze
phase transitions in the Ising model based on motif $M_1$. Unlike to
\cite{HinczBerker2006} in our case the Ising model has no phase
transition for $p=0$, which manifests the difference between our
construction and that used in \cite{HinczBerker2006}. We also show
that, for $L\leq L_*$, the model is in an unordered state whenever
$K=0$, i.e., the spin-spin interactions along the nonrandom bonds is
absent.  For $L>L_*$ and $p\geq p_*$, the model is in an ordered
state even for $K=0$. We plan to study the phase diagram of this
model in the $(K,h)$-plane in a separate work, where we also plan to
consider such problems for the graphs based on the remaining motifs.

\vskip.2cm \noindent \textbf{Acknowledgment:} This work was
supported in part by the DFG through the SFB 701: `Spektrale
Strukturen und Topologische Methoden in der Mathematik' and by the
European Commission under the project STREVCOMS PIRSES-2013-612669.

\section{Appendix}

Here we give detailed calculations of the quantities from subsection
\ref{SS24}. First we get the quantity in (\ref{nK1}):
\begin{eqnarray}
  \label{nk1}
 \langle n_k \rangle & = &  \sum_{l=1}^{k-1} \sum_{\nu =0}^{2(q-1)(l-1)} ( 2(q-1) + \nu) \frac{(q-1)\cdot q^{-l}}{ 1 + q^{1-k}}\cdot \\[.17cm] \nonumber
&&\cdot{\genfrac{(}{)}{0pt}{}{2(q-1)(l-1)}{\nu}} p^\nu (1-p)^{2(q-1)(l-1)-\nu} + \\[.2cm]
& + & \sum_{\nu=0}^{(q-1) (k-1)} (q-1+ \nu) \frac{2}{q^{k-1} + 1}  \nonumber \\[.17cm]
&\times & {\genfrac{(}{)}{0pt}{}{(q-1)(k-1)}{\nu}} p^\nu (1-p)^{(q-1)(k-1)-\nu} =\nonumber \\[.2cm]
& = & \frac{2(q-1)(1-p)(1-q^{1-k})}{1+q^{1-k}}+ \frac{2p(kq^{1-k}(1-q)+q-q^{1-k})}{1+q^{1-k}}+\nonumber \\[.2cm] \nonumber
& + & \frac{2(q-1+(q-1)(k-1)p)}{q^{k-1}+1} =
\frac{q^{k-1}(2q-2+2p)-2p}{q^{k-1}+1}.
\end{eqnarray}
Next, for the quantity in (\ref{nKK}), we have
\begin{eqnarray*}
 \langle n_k \rangle & = &  \sum_{l=1}^{k-1} \sum_{\nu =0}^{4(l-1)} ( 4 + \nu) \frac{3 \cdot 4^{-l}}{ 1 + 4^{1-k}}\cdot {
\genfrac{(}{)}{0pt}{}{4(l-1)}{\nu}} p^\nu (1-p)^{4(l-1)-\nu}+ \nonumber\\[.2cm] & + & \sum_{\nu=0}^{2 (k-1)} (2+ \nu) \frac{2}{4^{k-1} + 1} \cdot {\genfrac{(}{)}{0pt}{}{2(k-1)}{\nu}} p^\nu (1-p)^{2(k-1)-\nu}= \\[.2cm] & = &
 4+\frac{4}{3}\left(p-\frac{3+2p}{4^{k-1}+1}\right). \nonumber
\end{eqnarray*}
Now we calculate the quantity in (\ref{nKK1})
\begin{eqnarray*}
 \langle n_k \rangle & = &  \sum_{l=1}^{k-1} \Big[
\sum_{\nu =0}^{3(l-1)} ( 3 + \nu) \frac{  4^{-l}}{ 1 + 4^{1-k}}\cdot
{
\genfrac{(}{)}{0pt}{}{3(l-1)}{\nu}} p^\nu (1-p)^{3(l-1)-\nu} +\nonumber\\[.2cm]
 & + & \quad \sum_{\nu =0}^{4(l-1)} ( 4 + \nu) \frac{ 4^{-l}}{ 1 + 4^{1-k}}\cdot {
\genfrac{(}{)}{0pt}{}{4(l-1)}{\nu}} p^\nu (1-p)^{4(l-1)-\nu} +\nonumber\\[.2cm]
& + & \quad \sum_{\nu =0}^{5(l-1)} ( 5 + \nu) \frac{ 4^{-l}}{ 1 +
4^{1-k}}\cdot {
\genfrac{(}{)}{0pt}{}{5(l-1)}{\nu}} p^\nu (1-p)^{5(l-1)-\nu} \Big]+\nonumber\\[.2cm] & + & \sum_{\nu=0}^{2 (k-1)} (2+ \nu) \frac{4}{4^k + 4} \cdot {\genfrac{(}{)}{0pt}{}{2(k-1)}{\nu}} p^\nu (1-p)^{2(k-1)-\nu}+ \\[.2cm]
& + & \sum_{\nu=0}^{3 (k-1)} (3+ \nu) \frac{2}{4^k + 4} \cdot {\genfrac{(}{)}{0pt}{}{3(k-1)}{\nu}} p^\nu (1-p)^{3(k-1)-\nu} +\nonumber \\[.2cm]
& + &  \sum_{\nu=0}^{ k-1} (1+ \nu) \frac{2}{4^k + 4} \cdot {\genfrac{(}{)}{0pt}{}{k-1}{\nu}} p^\nu (1-p)^{k-1-\nu} \nonumber =\\[.2cm]
& = &
 4+\frac{4}{3}\left(p-\frac{3+2p}{4^{k-1}+1}\right). \nonumber
\end{eqnarray*}

\end{document}